\documentclass{article}

\newif\iffull
\fulltrue

\newif\ifabstract

\iffull
\abstractfalse
\else
\abstracttrue
\fi

\PassOptionsToPackage{numbers, sort&compress}{natbib}
\usepackage[preprint]{neurips_2023}

\usepackage[dvipsnames]{xcolor}

\usepackage{xspace}
\usepackage{enumerate,paralist}
\usepackage[utf8]{inputenc} 
\usepackage[T1]{fontenc}    
\usepackage{hyperref}       
\usepackage{url}            
\usepackage{booktabs}       
\usepackage{amsfonts}       
\usepackage{nicefrac}       
\usepackage{bm}
\usepackage{microtype}      
\usepackage{tikz}
\usepackage{enumerate}
\usepackage{amsmath,amsthm,amssymb}
\usepackage{url}
\usepackage{forest}
\usepackage{tikz}
\usepackage{wrapfig}
\usepackage{forest}
\usetikzlibrary{matrix,positioning,arrows.meta,arrows}
\usepackage{caption}
\usepackage{subcaption}
\usepackage{enumitem}
\usepackage{listings}
\usepackage{color}
\usepackage{cleveref}

\newif\ifnotes
\notestrue
\ifnotes
\newcommand{\sam}[1]{{\ifnotes \textcolor{PineGreen}{\scriptsize Sam: {#1}} \fi}}
\newcommand{\shikha}[1]{\ifnotes {\noindent \scriptsize  \textcolor{blue} {Shikha: {#1}}} \fi{}}
\newcommand{\ben}[1]{{\ifnotes \scriptsize \textcolor{orange}{Ben: {#1}} \fi}}
\newcommand{\aidin}[1]{{\ifnotes \scriptsize \textcolor{red}{Aidin: {#1}} \fi}}

\else
\newcommand{\sam}[1]{}
\newcommand{\shikha}[1]{}
\newcommand{\ben}[1]{}
\newcommand{\aidin}[1]{}

\fi

\newcommand{\defn}{\emph}
\renewcommand{\epsilon}{\varepsilon}

\renewcommand{\tilde}{\widetilde}

\newcommand{\learnedLL}{learnedLLA\xspace}
\newcommand{\LearnedLL}{learnedLLA\xspace}

\newcommand{\LLAs}{LLAs\xspace}
\newcommand{\LLA}{LLA\xspace}
\newcommand{\PMAs}{LLAs\xspace}
\newcommand{\PMA}{LLA\xspace}
\newcommand{\PMAA}{LLA algorithm\xspace}

\newcommand{\pma}{list labeling array\xspace}
\newcommand{\pmas}{list labeling arrays\xspace}

\newcommand{\calP}{{\mathcal P}}
\newcommand{\calD}{{\mathcal D}}

\renewcommand{\defn}[1]{\emph{\textbf{#1}}}
\newtheorem{theorem}{Theorem}
\newtheorem{lemma}[theorem]{Lemma}

\newtheorem{definition}[theorem]{Definition}
\newtheorem{corollary}[theorem]{Corollary}

\definecolor{dkgreen}{rgb}{0,0.6,0}
\definecolor{gray}{rgb}{0.5,0.5,0.5}
\definecolor{mauve}{rgb}{0.58,0,0.82}

\let\originaleta\eta
\renewcommand{\eta}{\bm{\originaleta}}

\title{Online List Labeling with Predictions}

\author{%
 Samuel McCauley \\
 Department of Computer Science\\
 Williams College\\
 Williamstown, MA 01267 \\
 \texttt{sam@cs.williams.edu} \\
   \And
   Benjamin Moseley \\
   Tepper School of Business \\
   Carnegie Mellon University\\
   Pittsburgh, PA 15213 \\
   \texttt{moseleyb@andrew.cmu.edu} \\
   \And
   Aidin Niaparast \\
   Tepper School of Business \\
   Carnegie Mellon University\\
   Pittsburgh, PA 15213 \\
   \texttt{aniapara@andrew.cmu.edu} \\
   \And
   Shikha Singh \\
 Department of Computer Science\\
 Williams College\\
 Williamstown, MA 01267 \\
   \texttt{shikha@cs.williams.edu} 
}

\begin{document}
\maketitle

\begin{abstract}
A growing line of work shows how learned predictions can be used to break through worst-cast barriers to improve the running time of an algorithm. However, incorporating predictions into data structures with strong theoretical guarantees remains underdeveloped.  This paper takes a step in this direction by showing that predictions can be leveraged in the fundamental \emph{online list labeling problem}. In the problem, $n$ items arrive over time and must be stored in \emph{sorted order} in an array of size $\Theta(n)$.  The array slot of an element is its \emph{label} and the goal is to maintain sorted order while minimizing the total number of elements moved (i.e., relabeled). We design a new list labeling data structure and bound its performance in two models.  In the worst-case learning-augmented model, we give guarantees in terms of the error in the predictions.  Our data structure provides strong theoretical guarantees---it is optimal for \emph{any prediction error} and guarantees the best-known worst-case bound even when the predictions are entirely erroneous. We also consider a stochastic error model and bound the performance in terms of the expectation and variance of the error. Finally, the theoretical results are demonstrated empirically.  In particular, we show that our data structure performs well on numerous real datasets,  including temporal data sets where predictions are constructed from elements that arrived in the past (as is typically done in a practical use case).
\end{abstract}

\section{Introduction}\label{sec:intro}

A burgeoning recent line of work has focused on coupling machine learning with discrete optimization algorithms.  The area is known as {\em algorithms with predictions}, or alternatively,  {\em learning augmented algorithms}~\cite{mitzenmacher2022algorithms,roughgarden2021beyond}.   
This area has developed a  framework for beyond-worst-case analysis that is generally applicable.  In this framework  an algorithm is given a prediction that can be erroneous. The algorithm can use the prediction to tailor itself to the given problem instance and the performance is bounded in terms of the error in the prediction.  

Much prior work has focused on using this framework in the online setting where learned predictions are used to cope with uncertainty \cite{LattanziLMV20}.  This framework has been further used for warm-starting offline algorithms to
improve the beyond worst-case running times of 
combinatorial optimization problems. This includes results on weighted bipartite matching \cite{DinitzILMV21}, maximum flows \cite{DaviesMoVa23},  shortest-paths \cite{ChenSVZ22} and convex optimization \cite{SakaueO22}.

 A key question  is to develop a theoretical understanding of how to improve the performance of data structures using learning.  
 Kraska et al.~\cite{KraskaBCDP18}  jump-started this area by using learning to improve the performance of indexing problems.   
 Follow up work on 
 learned data structures (e.g.~\cite{vaidya2020partitioned,lin2022learning})
 have demonstrated their advantage over traditional worst-case variants both empirically and through improved bounds when the learned input follows a particular distribution; see Section~\ref{sec:related}.

 The theoretical foundation of using learned predictions in data structures remains underdeveloped. In particular,   there are no known analysis frameworks showing how the performance of learned data structures degrade as a function of the error in prediction.  In this paper,  we take a step in this direction by showing how to leverage the learning-augmented algorithms framework for the fundamental \defn{online list labeling problem}. We
 bound its performance in terms of the error in the prediction and show that this performance is optimal for any error.

In the online list labeling problem\footnote{We use `online list labeling' and `list labeling' interchangeably;  the problem is trivial in the offline setting.}, the goal is to efficiently store a list of \emph{dynamic} elements in \emph{sorted order}, which is a common primitive in many database systems.   More formally, a set of $n$ elements arrive over time and must be stored in sorted order in an array of size $cn$, where $c > 1$ is a constant.  The array slot of an element is referred to as its \emph{label}. When a new element arrives, it must be stored in the array between its predecessor and successor which may necessitate shifting (relabeling) other elements to make room.  The goal of a list labeling data structure is to  minimize the number of elements moved; we refer to the number of element movements as the \defn{cost}.  The challenge is that the insertion order is online and adversarial, and the space is limited. A greedy approach, for example, might end up moving $\Omega(n)$ elements each insert.

List labeling algorithms have been studied extensively for over forty years,  sometimes under different names such as sparse tables~\cite{ItaiKoRo81}, packed memory arrays~\cite{BenderDeFa00,BenderHu07}, sequential file maintenance~\cite{Willard82Maintaining,willard1981inserting,Willard86Good,Willard92A}, and order maintenance~\cite{Dietz82}. We call any data structure that solves the list labeling problem a \defn{list labeling array (LLA)}.   Itai et al. \cite{ItaiKoRo81} gave the first LLA which moves $O(\log^2 n)$ elements on average over all inserts.  This has since been extended to guarantee  $O(\log^2 n)$ elements are moved on any insert  \cite{Willard82Maintaining}.  These results are  tight as there is a $\Omega(\log^2 n)$ lower bound on the amortized number of movements needed for any deterministic algorithm~\cite{BulanekKoSa12}. 

  The \LLA is a basic building block in many applications, such as
 cache-oblivious data structures~\cite{BenderFiGi05},
 database indices~\cite{raman1999locality} and joins~\cite{khayyat2017fast},
 dynamic graph storage~\cite{wheatman2021parallel,de2021teseo, pandey2021terrace},
 and neighbor search in particle simulations~\cite{durand2012packed}. Recently, \PMAs have also been used in the design of dynamic learned indices~\cite{ding2020alex}, which we describe in more detail in Section~\ref{sec:related}.

Due to the their wide use, several prior works have attempted to improve LLAs using predictions heuristically.  Bender et al.~\cite{BenderHu07} introduced the \defn{adaptive packed memory array  (APMA)} which tries to adapt to non-worst case input sequences. The APMA uses a heuristic predictor to predict future inserts.  They show guarantees only for specific input instances, and leave open the question of improving performance on more general inputs. Follow up data structures such as rewired PMA~\cite{de2019packed} and gapped arrays~\cite{ding2020alex} also touch on this challenge---how to adapt the performance cost of a \LLA algorithm to non-worst case inputs, while still maintaining $O(\log^2 n)$ worst case.

A  question that looms is whether one can extract improved performance guarantees using learned predictions on general inputs.  An ideal algorithm can leverage accurate predictions to  give much stronger guarantees than the worst case. Further, the data structure should maintain an $O(\log^2 n)$ guarantee  in the worst case even with completely erroneous predictions.  Between, it would ideal to have a natural trade-off between the quality of predictions and the performance guarantees. 

\subsection{Our Contributions}\label{sec:results}

In this paper, we design a new data structure for list labeling which  incorporates learned predictions and bound its performance by the error in the predictions. Our theoretical results show that high-quality predictions allow us to break through worst-case lower bounds and achieve asymptotically improved performance. As the error in the prediction grows, the data structure's performance degrades proportionally, which is  
never worse than the state-of-the-art LLAs.  Thus, the algorithm gives improvements with predictions at no cost. This is established in both worst-case and stochastic models.  Our experiments show that these theoretical improvements are reflected in practice.

\paragraph{Prediction model.} We define the model to state our results formally. Each element $x$ has a \defn{rank} $r_x$ that is equal to the number of elements smaller than it after all elements arrive.   When an element arrives, a \defn{predicted rank} $\tilde{r}_x$ is given.
The algorithm can use the predicted ranks when storing elements and the challenge is that the predictions may be incorrect.  Define  $\eta_x = |\tilde{r}_x - r_x|$ as $x$'s \defn{prediction error}.  Let $\eta = \max_x \eta_x$ be the \defn{maximum error} in the predictions.  

\paragraph{Upper bounds.}  In Section~\ref{sec:learnedpma}, we present a new data structure, called a \defn{\LearnedLL}, which uses any list labeling array as a black box.  We prove that it has the following theoretical guarantees.
 \begin{compactitem}
\item  Under arbitrary predictions, the {\LearnedLL} guarantees an amortized cost of $O(\log^2 \eta)$ per insertion when using a classic LLA~\cite{ItaiKoRo81} as a black box. This implies that the algorithm has the best possible $O(1)$ cost when the predictions are perfect (or near perfect). Even when $\eta=n$, the largest possible prediction error, the cost is $O(\log^2 n)$, matching the lower bound on deterministic algorithms.  Between, there is graceful degradation in performance that is logarithmic in terms of the error.   In fact, these results are more general.  Given a list labeling algorithm (deterministic or randomized) with an amortized cost of $C(n)$, for 
a reasonable function $C$\footnote{For admissible cost functions $C(n)$ as defined in Definition~\ref{def:admissible}.}%
, our data structure guarantees performance $C(\eta)$. This implies the new data structure can improve the performance of any known LLA.  See Section~\ref{sec:arbitrary_predictions}.
\item We also analyze the \LearnedLL when the error $\eta_x$ for each element $x$ is sampled independently from an \emph{unknown} probability distribution $D$ with average error $\mu$ and variance $s^2$.  Under this setting, we show that the amortized cost is 
$O( \log^2 (|\mu|+s^2))$
in expectation.
Interestingly, the arrival order of elements can be adversarial and even depend on the sampled errors.  This implies that for any distribution with constant mean and variance, we give the best possible $O(1)$ amortized cost. See Section~\ref{sec:stochastic_predictions}.
\end{compactitem}

\paragraph{Lower bound.} 
A natural next question is whether the \learnedLL is the best possible or if there is an LLA that uses the predictions in a way that achieves a stronger amortized cost. In Section ~\ref{sec:lowerbound}, we show our algorithm is the best possible: for any $\eta$, \emph{any} deterministic algorithm must have amortized cost at least $\Omega(\log^2 \eta)$, matching our upper bound.  Thus, our algorithm utilizes predictions optimally.

\paragraph{Empirical results.}  Finally, we show that the theory is predictive of empirical performance in Section~\ref{sec:experiments}.   We demonstrate on real data sets that: (1) the \learnedLL outperforms  state-of-the-art baselines on numerous instances, (2) when current \PMAs perform well, the \learnedLL matches or even improves upon their performing, indicating minimal overhead from using predictions, (3) the \learnedLL is robust to 
 prediction errors, (4) these results hold on real-time-series data sets where the past is used to make predictions for the future (which is typically the use case in practice), and (5) we demonstrate that a small amount of data is needed to construct useful  predictions.
 
\subsection{Related Work}\label{sec:related}

See \cite{roughgarden2021beyond} for an overview of learning-augmented algorithms in general.  
\iffull
Below we describe related work on learned data structures and list labeling data structures.
\fi

\paragraph{Learned replacements of data structures.} 
The seminal work on learning-augmented algorithms by Kraska et al.~\cite{KraskaBCDP18}, and several follow up papers~\cite{ding2020alex,wu2021updatable,marcus2020cdfshop,kipf2019sosd,kipf2020radixspline,ferragina2020learned,galakatos2019fiting} focus on designing and analyzing a learned index, which  replaces a traditional data structure like a $B$-tree and outperforms it empirically on practical datasets.  
Interestingly, the use of learned indices directly motivates the learned online list labeling problem.  To apply learned indices on dynamic workloads, it is necessary to efficiently maintain the input in sorted order in an array.  Prior work~\cite{ding2020alex} attempted to address this through a greedy list labeling structure, a \emph{gapped array}. A gapped array, however, can easily incur $\Omega(n)$ element movements per insert even on non-worst case inputs. This bottleneck does not manifest itself in the theoretical or empirical performance of learned indices so far as the input is assumed to be randomly shuffled. \iffull%
\footnote{In their analysis showing the effectiveness of learned indices, Ferragina et al.~\cite{ferragina2020learned} further assume that gaps between input elements (taken in sorted order) follow a specific distribution.}  
\fi
Note that random order inserts are the {\em best case} for list labeling and incur $O(1)$ amortized cost in expectation~\cite{bender2006insertion}.  
In contrast, our guarantees hold against inputs with adversarial order that can even depend on the prediction errors.

Besides learned indices replacing $B$-trees, other learned replacements of data structures include hash tables~\cite{sabek2022can,ferragina2023learned} and rank-and-select structures~\cite{boffa2022learned,ferragina2021performance,ferragina2021repetition}.

\paragraph{Learned adaptations of data structures.}
In addition to approaches that use a trained neural net to replace a search tree or hash table, several learned variants of data structures which directly adapt to predictions have also been designed.  These include learned treaps~\cite{lin2022learning, chen2022power}, filters~\cite{vaidya2020partitioned,mitzenmacher2018model,bercea2022daisy,wheatman2018packed} and count-min sketch~\cite{hsu2019learning,du2021putting}.  The performance bounds of most of these data structures assume perfect predictions; when robustness to noise in predictions is analyzed (see e.g.~\cite{lin2022learning}), the resulting bounds revert to the worst case rather than degrading gracefully with error.  The \LearnedLL is unique in the landscape of learned data structure as it guarantees (a) optimal bounds for any error, and (b) best worst-case performance when the predictions are entirely erroneous.

\paragraph{Online list labeling data structures.}  The online list labeling problem is described in two different but equivalent ways:  storing a dynamic list of $n$ items in a sorted array of size $m$; or assigning labels from $\{1, \ldots, m\}$ to these items where for each item $x$, the label $\ell(x)$ is such that $x < y \implies \ell(x) < \ell(y)$. We focus on the {\em linear} regime of the problem where $m = cn$, and $c>1$ is a constant. 

If an \PMA maintains that any two elements have $\Theta(1)$ empty slots between them, it is referred to as a \defn{packed-memory array} (PMAs)~\cite{BenderDeFa00}. PMAs are used as a subroutine in many algorithms:  e.g. applied graph algorithms~\cite{de2019fast,de2021teseo,wheatman2021parallel, pandey2021terrace,wheatman2021streaming,wheatman2018packed} and indexing data structures~\cite{BenderDeFa00,ding2020alex,kopelowitz2012line}. 

A long standing open question about \LLAs---whether randomized LLA algorithms can perform better than deterministic was resolved in a recent breakthrough.  In particular, Bender et al.~\cite{BenderCoFa22} extended the \emph{history independent} LLA introduced in~\cite{BenderBeJo16} and showed that it guarantees $O(\log^{3/2} n)$ amortized expected cost. 
The \learnedLL achieves $O(\log^{3/2} \eta)$ amortized cost when using this \PMA as a black box inheriting its strong performance; see Corollary~\ref{cor:cost}.

 Bender and Hu~\cite{BenderHu07} gave the first beyond-worst-case \LLA, the \defn{adaptive PMA (APMA)}.  The APMA has $O(\log n)$ amortized cost on specialized sequences: \emph{sequential} (increasing/ decreasing), \emph{hammer} (repeatedly inserting the predecessor of an item) and \emph{random} inserts.  APMA also guarantees $O(\log^2 n)$ amortized cost in the worst case. 
 Moreover, the APMA uses heuristics to attempt to improve performance on inputs outside of these specialized sequences. %
 \iffull
 (Similar heuristics were used in the \emph{Rewired Memory Array}~\cite{de2019packed}). %
 \fi
 These heuristics do not have theoretical guarantees, but nonetheless, the APMA often performs better than a classic one empirically~\cite{BenderHu07,de2019fast}.  While the \learnedLL bounds (based on prediction error) are incomparable with that of the APMA, our experiments show that the \learnedLL outperforms the APMA on numerous datasets.  In fact, if we combine our techniques by using an adaptive PMA as the black-box \PMA of the \LearnedLL, it performs better than both of them, reinforcing its power and versatility.

\section{Preliminaries}\label{sec:prelim}

 In this section we formally define the list labeling problem, our prediction model, and the classic list labeling data structures we use as a building block in Section~\ref{sec:learnedpma}.

\paragraph{Problem definition.} In the \defn{online list labeling problem}, the goal is to maintain a dynamic sorted list of $n$ elements in an array of $m = cn$ slots for some constant $c > 1$.  
We refer to $n$ as the \defn{capacity} of the array, $m$ as the \defn{size}, and $n/m$ as the \defn{density}.
The \defn{label} of an element is the array slot assigned to it.  As new elements arrive, existing elements may have to be moved to make room for the new element.
All elements come from an ordered universe of possible elements $U$. The list labeling structure maintains that the elements in the array must be stored in sorted order---for any $x_i, x_j$ stored in the array, if $x_i < x_j$ then the label of $x_i$ must be less than the label of $x_j$.

We refer to 
a data structure for this problem as a \defn{\pma (\PMA)}.
An \defn{element movement} is said to occur each time an element is assigned a new label.\footnote{This means that every element has at least one element movement when it is inserted.}
We use the term element \emph{movements} and \emph{relabels} interchangeably. 
The \defn{total cost} of an \PMA is the total number of element movements.
The \defn{amortized cost} of an \PMA is the total cost during $k$ insertions, divided by $k$. This cost function is standard both in the theory and practice of \LLAs (see e.g.~\cite{ItaiKoRo81,BenderCoDe02,BenderCoFa22,BenderHu07,de2019packed}).

\paragraph{Data structure model.} We assume that the set $S$ of $n$ elements are drawn one-by-one adversarially from $U$ and inserted into the \LLA.  We use existing \PMAs~(e.g.~\cite{BenderCoFa22,BenderBeJo16,BenderCoDe02,ItaiKoRo81}) as a black box in our data structure.  We consider any (possibly randomized) \LLA $A$ which supports the following operations:
\begin{compactitem}
    \item \textsc{Insert}$(x)$:  Inserts $x$ in $A$.  The slot storing $x$ must be after the slot storing its predecessor, and before the slot storing its successor.
    \item \textsc{Init}$(S')$:  Given an empty LLA $A$ and given a (possibly-empty) set $S' \subseteq U$, insert all elements from $S'$ in $A$.
\end{compactitem}

Note that \textsc{Init}$(S')$ can be performed using $O(|S'|)$ element movements---as all $|S'|$ elements are available offline, they can be placed in $A$ one by one.

For simplicity we assume that $n$ is a power of $2$; our results easily generalize to arbitrary $n$.  All logarithms in this paper are base 2.

In this work, we do not explicitly discuss deletes.  However, our data structure can be easily extended to support deletes as follows.  We build a \LearnedLL with capacity $2n$.
We split all operations into \defn{epochs} of length $n$.  During an epoch we ignore all deletes.  After the epoch completes, we rebuild the data structure, removing all elements that were deleted during the epoch.  This costs $O(n)$, so the amortized cost is increased by $O(1)$.

We parameterize the amortized cost of the \LearnedLL by the amortized cost function of \textsc{Insert} for a black-box \PMAA of density $1/2$.\footnote{If we initialize a \PMA $P_j$ with $m'$ slots, we never store more than $m'/2$ elements in $P_j$.  However, this choice is for simplicity, and our results immediately generalize to \PMAs using different parameter settings.}  For background on how a classic \PMA, the packed-memory array (PMA)\cite{ItaiKoRo81, BenderDeFa00}, works see 
\iffull
Appendix~\ref{sec:appendix}.
\else
the full version  of the paper.  
\fi
 Let $C(n')$ upper bound the amortized cost to insert into a \PMA with capacity $n'$ and size $2n'$.  
 \begin{definition}\label{def:admissible}
We say that a cost function $C(\cdot)$ is \defn{admissible} if:
\iffull
\begin{itemize}[noitemsep,nolistsep]
    \item for all $i,j$ with $i < 2j$, $C(i) = O(C(j))$,
    \item $C(n')  = \Omega(\log n')$,   and 
    \item for all $j$, $\sum_{i=1}^{\infty} C(2^{j + i})/2^i = O(C(j))$.   
\end{itemize}
\else
    (1) for all $i,j$ with $i < 2j$, $C(i) = O(C(j))$,
    (2) $C(n')  = \Omega(\log n')$,  
    and (3) for all $j$, $\sum_{i=1}^{\infty} C(2^{j + i})/2^i = O(C(j))$.   
\fi
\end{definition}

Both $C(n) = \Theta(\log^2 n)$ (the optimal cost function for any deterministic \PMA) and $C(n)  = \Theta(\log^{3/2} n)$ (the best-known cost function for any randomized \PMA~\cite{BenderCoFa22}) are admissible.  There is an unconditional lower bound that any \PMA must have amortized cost at least $\Omega(\log n')$~\cite{BulanekKaSa13}. 

\iffull
The prediction model is introduced in Section~\ref{sec:results}. Recall that
$\eta_x = |\tilde{r}_x - r_x|$ is the prediction error of element $x$, where $\tilde{r}_x$ is its \defn{predicted rank} and $r_x$ is its true (unknown) \defn{rank} in $S$.   
In Section~\ref{sec:arbitrary_predictions}, the predicted ranks may be assigned adversarially and we bound the cost in terms of the \defn{maximum error} $\eta = \max_i \eta_i$. 
 In Section~\ref{sec:stochastic_predictions}, we assume that the prediction errors are 
 drawn independently from an unknown distribution $\calD$ and bound the cost in terms of the mean and variance of $\calD$.
 \fi

\iffull
\section{Our Data Structure: the \LearnedLL}\label{sec:learnedpma}

In this section we design and analyze our data structure, the \LearnedLL which uses any generic \pma as a black box.  For space, many of the proofs in this section have been removed; they can be found in the full version, included in the supplementary materials.

We begin in Section~\ref{sec:datastructuredescription} by describing how the \LearnedLL works. In Section~\ref{sec:arbitrary_predictions} we give a worst-case analysis for the \LearnedLL.
If the black-box \PMA has admissible amortized cost $C(n)$ for $n$ insertions, our data structure achieves an amortized cost of $C(\eta)$ for maximum error $\eta$.  

In Section~\ref{sec:lowerbound}, we show that the \LearnedLL is {\em optimal} for deterministic solutions to list labeling.  Bul{\'a}nek et al. showed that that in the worst case any deterministic \PMAA must have worst-case cost $C(n') = \Omega(\log^2 n')$~\cite{BulanekKoSa12}. We use this lower bound to show that for any maximum error $\eta$, any deterministic \LLA with $m = O(n)$ slots must incur total cost $\Omega(n \log^2 \eta)$.

In Section~\ref{sec:stochastic_predictions} we give improved bounds for our data structure when the predictions are not adversarial but rather drawn from a distribution with bounded expectation and variance.  If the distribution has expectation $\mu$ and variance $s^2$, then using the best-known randomized list labeling data structure gives expected amortized cost $O(\log^{3/2}(\mu + s))$.

\paragraph{Learned List-Labeling Array.}  
Our learned list labeling data structure partitions its array into several subarrays, each of which is maintained as its own \pma.  At a high level, when an element $x_i$ is inserted along with its predicted rank $\tilde{r}_i$, we use $\tilde{r}_i$ to help determine which \PMA to insert into while maintaining sorted order. 
Intuitively, the performance comes from keeping these constituent \pmas small: 
if the predictions are good, it is easy to maintain small \LLAs while retaining sorted order.  But, if some element has large prediction error, its \pma grows large, making inserts into it more expensive.  

Below, we describe and analyze this data structure more formally.

\subsection{Data Structure Description}
\label{sec:datastructuredescription}

The \LearnedLL with capacity $n$ operates on an array of $m = 6n$ slots.  At all times, the \LearnedLL is partitioned into $\ell$ \defn{actual \PMAs} $P_1, P_2, \ldots P_\ell$ for some $\ell$.  In particular, we partition the $m$ array slots into $\ell$ contiguous subarrays; each subarray is handled by a black-box \PMA.  We maintain that for all $i < j$, all elements in $P_i$ are less than all elements in $P_j$.

We define a tree to help keep track of the actual \PMAs. Consider an implicit, static and complete binary tree $T$ over ranks $\{1,\ldots n\}$. Each node in $T$ has a set of \defn{assigned ranks} and \defn{assigned slots}.   Specifically, the $i$th node at height $h$ has $2^h$ assigned ranks $\{2^h (i-1) + 1, 2^h (i-1) + 2, \ldots, 2^h i\}$,
and $6\cdot 2^h$ assigned slots $\{2^h\cdot 6(i-1) + 1, 2^h \cdot 6(i-1) + 2, \ldots, 2^h\cdot 6i\}$.

We use each actual \PMA as a black box to handle how elements are moved within its assigned slots. 
 To obtain the \LearnedLL label for an element $x_i$ in actual \PMA $P_j$, we sum the black-box label of $x_i$ in $P_j$ and the value of the smallest slot assigned to $P_j$ minus $1$.

Every actual \PMA has a set of assigned ranks and assigned slots; these ranks and slots must be from some node in $T$.  We say that the \PMA \defn{corresponds} to this node in $T$.
Each root-to-leaf path in $T$ passes through exactly one node that corresponds to an actual \PMA.
This means that the assigned ranks of the actual \PMAs partition $\{1, \ldots, n\}$ into $\ell$ contiguous subsets, and the assigned slots of the actual \PMAs partition $\{1, \ldots, m\}$ into $\ell$ contiguous subsets.

If a node $v$ in $T$ is not assigned to an actual \PMA, it may be useful to consider the assigned ranks and slots that would be used if a \PMA assigned to $v$ were to exist.  We call such a \PMA (that is not in $P_1,\ldots, P_\ell$) a \defn{potential \PMA}.
 The \defn{sibling}, \defn{parent} and \defn{descendants} of a \PMA $P_j$ refer to the LLAs corresponding to the sibling, parent and descendant nodes respectively of $P_j's$ node in the tree.

For any potential \PMA $P_j$, 
let ${|P_j|}$ denote the number of assigned ranks of $P_j$.
Initially $\ell=n$ and $P_i$ is assigned to the $i$th leaf of $T$, where $1 \leq i \leq n$---that is to say, the actual \PMAs begin as the leaves of $T$, each with one assigned rank.  The parameter $\ell$ changes over time as elements are inserted.  

Each element is inserted into exactly one actual \PMA. 
If an \PMA has density more than $1/2$, 
we \defn{merge} it with the (potential or actual) \PMA of its sibling node in $T$. 
 
\paragraph{Insertion Algorithm.}

To insert an item $x$ into the \LearnedLL, let $i_p$ and $i_s$ be the index of the \PMA containing the predecessor and successor of $x$ respectively.  If the successor of $x$ does not exist, let $i_p=\ell$ (the last actual \PMA); if the predecessor does not exist let $i_p = 1$.
Finally, let $i_x$ be the \PMA whose assigned ranks contain $\tilde{r_x}$.  
Call $\textsc{Insert}(x)$ on \PMA $P_i$ where:

\begin{equation}
i =
\left\{
	\begin{array}{ll}
		i_p  & \text{if } i_p > i_x \\
		i_s & \text{if } i_s < i_x \\
            i_x & \text{otherwise}
	\end{array}
\right.
\end{equation}

The black-box LLA $P_i$ updates its labels internally, and these updates are mirrored in the \LearnedLL.

If the insert causes the density of $P_i$ to go above $1/2$,  we \defn{merge} it with its sibling \PMA. 
Specifically, let $P_p$ be the parent of $P_i$; it must currently be a potential \PMA.  We take all actual \PMAs that are descendants of $P_p$, and merge them into a single actual \PMA.
Let $S_p$ be the set of elements currently stored in the slots assigned to $P_p$ in the \LearnedLL; thus, $S_p$ consists exactly of the contents of all actual \PMAs that are descendants of $P_p$.  We call $\textsc{Init}(S_p)$ on $P_p$, after which $P_p$ is an actual \PMA; all of its descendants are no longer actual \PMAs.
Note that after a merge, the assigned ranks and assigned slots of the actual \PMAs still partition $\{1, \ldots, n\}$ and $\{1, \ldots, m\}$ respectively.

After we insert into $P_i$, and merge if necessary, the \LearnedLL insert is complete.

Next, we analyze the performance of the \LearnedLL.  Let $C(\cdot)$ be amortized cost of the black-box \PMA  used in the \LearnedLL---our results will be in terms of $C(\cdot)$.\footnote{We assume there is one type of \PMA; otherwise, let $C(\cdot)$ upper bound the cost of all \PMAs being used.}  We assume that $C(\cdot)$ is admissible (Definition~\ref{def:admissible}).

\subsection{Arbitrary Predictions}\label{sec:arbitrary_predictions}
In this section, we assume that the input sequence and associated predictions can be arbitrary.
The adversary chooses elements one by one from a universe $U$. For each element, the adversary assigns a predicted rank (which can be based on past insertions, past predictions, and even the 
\LearnedLL algorithm being used) and inserts the element.
We show that the \LearnedLL achieves amortized cost $O(C(\eta))$, where $\eta$ is the maximum error.  

Our proof begins with several structural results, from which the analysis follows almost immediately.  

First, we show that to determine the overall asymptotic cost it is sufficient to consider the set of actual \PMAs after all inserts are completed.  That is,  the relabels that occur during merges and inserts to smaller \PMAs are lower-order terms.

\begin{lemma}
    \label{lem:finalpmas}
    For any sequence of insertions, 
    if $\calP_F$ is the set of actual \PMAs after all operations are completed, then the total number of element movements incurred by the 
    \LearnedLL is
    \[
    O\left( \sum_{P\in \calP_F} |P| \cdot C(|P|) \right).
    \]
\end{lemma}
\begin{proof}
For a \PMA $P\in \calP_F$, consider all potential \PMAs that are descendants of $P$ in $T$.  This partitions the merges: every $P'$ that is an actual \PMA at any time during $\sigma$ must be a descendant of exactly one $P\in\calP_F$.

First, we analyze the cost of all merges.  
The total capacity of all \PMAs that are a descendant of $P$ is $|P|\log |P|$.  Since a merge operation has linear cost, the total number of element movements during all merges is $O(|P|\log |P|)$.

Now, we must bound the cost of all inserts.  
Consider all inserts into all \PMAs $P_1, P_2, \ldots, P_d$ that are a descendant of $P$.  Let $k_i$ be the number of inserts into some such \PMA $P_i$; the total cost of these inserts is $O(k_i C(|P_i|)) \leq O(k_i C(|P|))$.  All inserts into descendants of $P$ are ultimately stored in $P$; thus, $\sum_{i=1}^d k_i \leq 3|P|$.  Then the total cost of all inserts into \PMAs that are a descendant of $P$ is $\sum_{i=1}^d O(k_i C(|P|)) = O(|P| C(|P|))$.

Summing between the merge and insert cost, and using $C(|P|) + \log |P| = O(C(|P|))$ because $C(\cdot)$ is admissible, we obtain a total number of element movements of 
\[
O\left( \sum_{P\in \calP_F} |P| \cdot C(|P|) + |P|\log |P|\right) \leq 
O\left( \sum_{P\in \calP_F} |P| \cdot C(|P|) \right).
\]
 \end{proof}

 \begin{lemma}
     \label{lem:higherrorelement}
     At any time, if $P$ is an actual \PMA, then there exists an element $x_j$ stored in $P$ with $\eta_j \geq |P|/2$.
 \end{lemma}
\begin{proof}
Since $P$ was formed by merging its two children $P_i$ and $P_k$ in $T$, the density of either $P_i$ or $P_k$ must have been at least $1/2$; without loss of generality assume it was  $P_i$.
       
	Let $r_1, r_1 + 1, \ldots, r_2$ be the sequence of $|P_i|$ ranks assigned to $P_i$.  Let $x_s$ and $x_\ell$ be the smallest and largest items in $P_i$ respectively, with predicted ranks $\tilde{r_s}$ and $\tilde{r_\ell}$.  We must have $\tilde{r_s} \geq r_1$---since the predecessor of $x_s$ is not in $P_i$, $x_s$ must be placed in the \PMA whose assigned ranks contain $\tilde{r_s}$, or in the \PMA containing its successor (if $\tilde{r_s}$ is larger than any rank assigned to the \PMA containing its successor).  
    Similarly, $\tilde{r_\ell} \leq r_2$.
 Therefore, $\tilde{r_\ell} - \tilde{r_s} \leq |P_i|- 1$. There are at least $3|P_i|$ items in $P_i$ as its density is at least $1/2$.  Thus, $r_\ell - r_s \geq 3|P_i| - 1$.  Thus, either $|\tilde{r_s} - r_s| \geq |P_i|$ or $|\tilde{r_\ell} - r_\ell| \geq |P_i|$.  
 Noting that $|P_i| = |P|/2$, 
 either $\eta_s \geq |P|/2$ or $\eta_\ell \geq |P|/2$.
\end{proof}

 \begin{theorem}
\label{lem:costtoinsert}
For any sequence of at most $n$ insertions $\sigma$ with maximum error $\eta$, 
the \LearnedLL using \PMAs with admissible cost function $C(\cdot)$
incurs $O(n C(\eta))$ total element movements.
\end{theorem}
\begin{proof}
By Lemma~\ref{lem:higherrorelement}, any $P\in\calP_F$ must have an element with error at least $|P|/2$.  Thus, any $P\in \calP_F$ must have $|P| \leq 2\eta$.
 The actual \PMAs partition the ranks, so $\sum_{P\in \calP_F} |P| = n$.  By Lemma~\ref{lem:finalpmas}, and 
 since $C(\cdot)$ is admissible,
 the total number of element movements over $\sigma$ is
\[
O\left(\sum_{P\in \calP_F} |P| \cdot C(|P|) \right)\leq 
O\left(\sum_{P\in \calP_F} |P| \cdot C(2\eta)\right) \leq 
O(n C(\eta)).\qedhere
\]
\end{proof}

\begin{corollary}\label{cor:cost}
Using the classic optimal LLA of Itai et al.~\cite{ItaiKoRo81}, the \learnedLL is a deterministic list labeling data structure with with $O(\log^2 \eta)$ amortized element movements.  Using the best-known randomized \PMA of Bender et al.~\cite{BenderCoFa22},  the \learnedLL is a randomized list labeling data structure with $O(\log^{3/2} \eta)$ amortized element movements.
\end{corollary}

\subsection{Optimality for Deterministic Learned List Labeling}\label{sec:lowerbound}

One may wonder if the \LearnedLL is optimal.  The \learnedLL is clearly optimal on the two extremes:  perfect predictions ($\eta = 0$) and completely erroneous predictions ($\eta=n)$.  When predictions are perfect, the \learnedLL does not need to move any element. When $\eta = n$, it is easy to create instances where the predictions give no information about the rank of each item (e.g.\ if each prediction is $\tilde{r_i} = 1$).
Bul\'{a}nek et al.~\cite{BulanekKoSa12} showed that, without predictions, any deterministic \PMA with $O(n)$ slots can be forced to perform $\Omega(\log^2 n)$ amortized element movements. Thus, the \learnedLL is optimal when $\eta=n$.

However, it is not immediately clear if the above idea extends to intermediate error. In this section, we show that for deterministic learned labeling data structures, the \LearnedLL is in fact optimal for all $\eta$.  
In particular, we show that for any $\eta$, there is an adversary that inserts elements with error at most $\eta$ that causes any deterministic LLA algorithm to incur $\Omega(n\log^2 \eta)$ total element movements. 

The intuition behind our proof is to split the learned list labeling problem into a sequence of $\Omega(n/\eta)$ subproblems of size $\eta$, each handled by its own \PMAA.
The given predictions within each subproblem are only accurate to within $\eta$---therefore, an adversary can force each \PMA to get \emph{no} information about an element's position within the \PMA, and by~\cite{BulanekKoSa12} perform no better than $\Omega(\log^2 \eta)$ amortized element movements.

The challenge is that the \PMAs for the subproblems are not actually separate:  we can't force the data structure to allocate $O(\eta)$ space to each \PMA and keep this allocation static throughout the execution.  The data structure may move elements between adjacent \PMAs or shift the number of slots available to each, an ability that disrupts a black-box lower bound.

\paragraph{Shifted list labeling.}
To address this issue, we describe a generalization of the list labeling problem.  We then show that this generalization has the same asymptotic cost as the classic list labeling problem, while being able to better handle the above issue where slots are dynamically allocated to a subproblem.
We call our generalization the \defn{shifted list labeling problem}.
 The shifted list labeling problem is similar to the online list labeling problem: $n$ elements from a totally-ordered universe arrive and must be assigned labels.  We assume these labels are integers without loss of generality.  
 At all times, the labels must respect sorted order: if $x < y$ then the label of $x$ must be less than the label of $y$.
 In the shifted list labeling problem labels are not necessarily in $\{1, \ldots, m\}$; instead, we require that the smallest and largest label differ by at most $m-1$.\footnote{Thus, the shifted list labeling problem generalizes the online list labeling problem: labels in $\{1, \ldots, m\}$ are sufficient, but not necessary, to satisfy this requirement.}  We call a data structure for this problem is a \defn{shifted \PMA}, and we call $m-1$ the \defn{maximum spread} of the shifted \PMA.

First, we show that shifted \PMAs have the same asymptotic cost as normal \PMAs.

\begin{lemma}
\label{lem:generalizedequivalent}
Assume that there exists a shifted \PMA $P^S$ of capacity $n$ with maximum spread $m-1$ that can handle a sequence of operations $\sigma$ with $T(n)$ total element movements.
Then there exists a \PMA $P^*$ of size $3m$ that can handle $\sigma$ with total element movements $2T(n)$.
\end{lemma}

\begin{proof}
Create a \PMA  $P^*$ of size $3m$.  Let $\ell_0$ be the label of the first item that is inserted in $P^S$.  We maintain that if an item $x$ has label $\ell_x$ in the shifted \PMA $P^S$, we give it label $\ell_x' = (\ell_x - \ell_0) + (m + 1)$ in $P^*$.  Every time an element $x$ is moved in $P^S$, the movement is mirrored in $P^*$, that is, $x$ is moved exactly once to set its label to $\ell_x'$.

We maintain the invariant that the label $\ell_x' \in \{1, \ldots, 3m\}$ for all $x$ in $P^*$. If $\ell_x$ for any item $x$ in $P^S$ is ever 
less than $\ell_0 - m$ or more than $\ell + 2m$ 
we \defn{rebuild}: we update $\ell_0$ to be the label of the smallest item currently in $P^S$, and recompute $\ell_x'$. That is, for all $x$ (with label $\ell_x$ in $P^S$), we set its label in $P^*$ to be $(\ell_x - \ell_0) + (m+1)$.  This requires moving each element once in $P^*$.

 We note that by construction, the difference between any two labels in $P^*$ is at most $m-1$.  We use this fact to bound the number element movements in $P^*$ below.
 
Consider the sequence of operations between two successive rebuilds $r$ and $r'$.  After $r$, all labels of elements in $P^*$ must be between $m$ and $2m$.  
Immediately before $r'$, there must be some element in $P^*$ with label less than $0$ or more than $2m$---but this means that no element in $P^*$ has label between $m$ and $2m$. This means that every element must have been moved at least once in $P^S$ between $r$ and $r'$.  
We move each element one more time during the rebuild $r'$.  This at most doubles the number of element movements.  Thus, there are $T(n)$ total element movements during all rebuilds; summing this and the $T(n)$ element movements that $P^*$ makes to mirror $P^S$, we obtain the lemma.
\end{proof}

The deterministic lower bound of Bul\'{a}nek et al.~\cite{BulanekKoSa12} is constructed using an adversary that builds the worst-case sequence of inserts one by one.  Our lower bound will use this approach as well.  
An \defn{adversary} for online list labeling maps from the labels of all currently stored items, to a new item to insert into the data structure.  
Repeatedly querying the adversary gives a fixed sequence of operations $\sigma^D$ for any deterministic \PMAA $D$. Let $\text{cost}(\sigma^D, D)$ be the total number of element movements incurred by $D$ on the operations in $\sigma^D$ generated by the adversary.
The goal of the adversary is to maximize $\min_D \text{cost}(\sigma^D, D)$.
Bul\'{a}nek et al.'s main result is an adversary that achieves $\min_D \text{cost}(\sigma^D, D) = \Omega(\log^2 k)$ for any deterministic \PMAA with capacity $k$.  For completeness, we summarize their result in Lemma~\ref{lem:optimallower}.  

\begin{lemma}[\cite{BulanekKoSa12}]
\label{lem:optimallower}
Define $\chi^k(m', |U|)$ to be the smallest number of element movements that the best adversary achieves for any deterministic \PMAA with capacity $k$ and $m'$ slots where inserts are from a universe $U$.
Then if 
$m' = (1 + \Theta(1))k$,
and $|U| \geq C'm$ for a sufficiently large constant $C'$, 
$\chi^k(m', |U|) = \Theta(k\log^2 k )$.
\end{lemma}

The goal of this section is to give an adversary such that for any deterministic learned list labeling data structure $D$ with $m = O(n)$ slots, the adversary forces $D$ to incur $\Omega(n\log^2 \eta)$ element movements, even when all items inserted by the adversary have predicted rank with error at most $\eta$. 
 Recall $c = m/n$; we assume that $c$ is a constant larger than $1$.

\paragraph{Lower bound setup.} Let $U$ be a universe of sufficient size for the lower bound in Lemma~\ref{lem:optimallower}.
We partition the universe $U'$ of size $|U|n/\eta$ into $n/\eta$ contiguous \defn{blocks} of size $|U|$.  We assume without loss of generality that $n/\eta$ is an integer (if not, we round $n$ up to the next multiple of $\eta$).
We call the $b$th block of $U$ \defn{block $b$} (i.e.\ we refer to blocks by an index $b\in \{1, \ldots n/\eta\}$).

The predicted rank of any $x$ in block $b$ is $b\eta$.
The adversary (described below) inserts exactly $\eta$ elements from each block in~\ref{thm:lower}; therefore, the error of any item $\eta_x \leq \eta$.

Let $S_t$ be the set of items stored in $D$ after the first $t-1$ operations (i.e.\ immediately before the $t$th operation).  We say that $S_t$ is the set of items stored \defn{at time $t$}.
For any block $b$ with at least two elements in $S_t$, let $\ell_b^t$ and $s_b^t$ be the largest and smallest label of any element in $S_t$ from block $b$.  Call $\ell_b^t - s_b^t$ the \defn{block segment size} of $b$ at time $t$.
We say that a block $b$ is \defn{active} at time $t$ if either it has less than two items at time $t$, or if its block segment size  is $\ell_b^t - s_b^t\leq 2c\eta$.
The adversary uses the current set of labels in the \PMA to determine the next item to insert.  The crux of our lower bound is that the adversary's $t$th insert is from a block that is active at time $t$.
This means that the set of items from any block $b$ form a shifted \LLA with maximum spread $2c\eta$.  This is formalized in Lemma~\ref{lem:blocksegment}.

\begin{lemma}\label{lem:blocksegment}
    Consider a block $b$, and consider a sequence of inserts $\sigma$ into $D$.  Let $\sigma^b_1, \sigma^b_2, \ldots, \sigma^b_{k}$ be the operations in $\sigma$ that insert an item from $b$. Assume that $b$ is active immediately before each $\sigma^b_i$. For $i\in \{1, \ldots, k\}$, let $L^b_i$ be the set of labels of elements from block $b$ immediately \emph{before} insert $\sigma^b_i$. Then, the following holds.
    \begin{enumerate}[label=(\alph*), noitemsep, leftmargin=*]
        \item \label{parta} The sequence $L^b_i$ for $i\in \{1, \ldots, k\}$ constitutes a shifted \PMA $A^S_b$ with maximum spread $2c\eta$. 
        \item \label{partb} There exists a \PMA $A_b$ with size $6c\eta+3$ 
        that moves at most twice as many elements as $A_b^{S}$ on the sequence of inserts $\sigma^b_1, \sigma^b_2, \ldots, \sigma^b_{k}$.  
        \item \label{partc}  Let $T(A_b)$ be the total number of element movements in $A_b$. Then, the total number of element movements in $D$ is lower bounded by  $\sum_{b=1}^{n/\eta} T(A_b)/2$. 
    \end{enumerate}
\end{lemma}

\begin{proof}

The labels in $L^b_i$ must be in sorted order by correctness of $D$.  Furthermore, the largest and smallest label in $L^b_i$ must differ by at most $2c\eta$ since block $b$ is active at operation $\sigma_i^b$.  Thus, we can construct a shifted \PMA $A^S_b$, such that the labels after insert $i$ are $L^b_{i+1}$. This proves Part~\ref{parta}.

Part~\ref{partb} follows from Lemma~\ref{lem:generalizedequivalent}.

To show Part~\ref{partc}, let $T(A^S_b)$ be the total number of element movements in $A^S_b$.  If an element's label changes in $A^S_b$, it must also change at least once in $D$.  Therefore, the total number of element movements in $D$ is lower bounded by $\sum_b T(A^S_b)$.  We have that $T(A_b) \leq 2T(A^S_b)$ by Lemma~\ref{lem:generalizedequivalent}.
\end{proof}

\paragraph{Adversary definition and final lower bound.} We define our adversary as follows.

First, at time $t$, the adversary arbitrarily chooses a block $b$ such that: (1) $b$ is active at time $t$, and (2) at most $\eta$ elements from $b$ are in $S_t$.
Let $A_b$ be the deterministic \PMA for block $b$ at time $t$ from Lemma~\ref{lem:blocksegment}.  The adversary of~\cite{BulanekKoSa12} from Lemma~\ref{lem:optimallower} (with capacity $k \triangleq \eta$ and $m \triangleq 6c\eta$ slots) maps the set of labels of items in $A_b$ at $t$ to a new item $x$ from block $b$. Then our adversary inserts $x$ with predicted rank $b\eta$.

If no block with less than $\eta$ elements is active at time $t$, the adversary performs the \defn{clean up step}:
for each block $b$ with less than $\eta$ elements,
the adversary inserts arbitrary elements from $b$ (with prediction $b\eta)$ until $\eta$ elements have been inserted from $b$.
After this is done, the adversary has inserted $n$ items.
The clean up step is to ensure that the predicted ranks are accurate to within $\eta$.

Now, we analyze this adversary.
First, we show that there are many blocks for which all items were inserted using the above adversary.

\begin{lemma}
\label{lem:mostblocksfull}
There are at least $n/(2\eta)$ blocks $b$ such that the adversary inserts at least $\eta$ items from $b$ before the cleanup step.
\end{lemma}
\begin{proof}
Let $t'$ be the time when the adversary enters the cleanup step.

Blocks are disjoint, so  at $t'$, the sum all block segment sizes is at most $m$ and the average block segment size is at most $c\eta$.
Therefore, at most $n/(2\eta)$ blocks have block segment size more than $2c\eta$; thus, at least $n/(2\eta)$ blocks have block segment size at most $2c\eta$.  Each of these at least $n/2\eta$ blocks must contain $\eta$ elements.
 \end{proof}

Now we are ready to complete the proof of our lower bound.

\begin{theorem}\label{thm:lower}
For any $\eta$, 
there exists an adversary 
that inserts $n$ items, each with error at most $\eta$, such that any deterministic data structure for learned list labeling with capacity $n$ and $m = cn$ slots incurs $\Omega(n\log^2 \eta)$ total element movements.
\end{theorem}
\begin{proof}
Let $B$ be the
set of 
blocks from Lemma~\ref{lem:mostblocksfull} that have $\eta$ inserted elements before the cleanup step; 
we have $|B| \geq n/2\eta$.

Consider a block $b\in B$.  Let $A_b$ be the \PMA for elements in $b$ from Lemma~\ref{lem:generalizedequivalent}.  All inserts of elements from $b$ are from the adversary in Lemma~\ref{lem:optimallower}; therefore, we must have $T(A_b) = \Omega(\eta \log^2 \eta)$.

Summing over the blocks in $B$, and using the lower bound on element movements in $D$ from Lemma~\ref{lem:blocksegment}, there must be $\Omega(n\log^2 \eta)$ element movements in $D$.
\end{proof}

\subsection{Stochastic Predictions with Adversarial Insert Order}\label{sec:stochastic_predictions}
In this section, we assume that each predicted rank is the result of adding random noise to the true rank.  In particular, given an \emph{unknown} distribution $\calD$, for any item $x_i$, we have $\tilde{r_i} = r_i + e_i$ where each $e_i$ is sampled independently from $\calD$.  This means the that error of each item is essentially sampled from $\calD$, as $\eta_i = |e_i|$.   If $\calD$ has \defn{expectation $\mu$} and \defn{variance $s^2$}, we show that under such predictions the \learnedLL has expected amortized cost $O(C((\mu + s^2)^2))$, where
$C(n)$ denotes an admissible cost function of its constituent \PMAs.

The arrival order of the elements can still be chosen adversarially, with full knowledge of the predicted ranks.  In particular, we can describe our input model with the following adversary. The adversary first chooses $n$ elements of the set $S$ from a totally ordered universe.  For each element $x_i \in S$, the error $e_i$ is then sampled from $\calD$.  Finally, the adversary can look at the predicted ranks and choose any insert order of elements in $S$.

First, we bound the probability an element achieves a certain error---this is useful because Lemma~\ref{lem:higherrorelement} shows that if $P$ is an actual \PMA, then 
it must contain an element $x_j$ with large $\eta_j$.
\begin{lemma}
 \label{lem:cheby}
 Assume that $\tilde{r_j} = r_j + e_j$ where $e_j$ is drawn independently from a distribution $\calD$ with expectation $\mu$ and variance $s^2$.
Then $\Pr[\eta_j \geq k(|\mu| + s^2)] \leq 1/k^2$.
\end{lemma}
\begin{proof}
Consider an element $x_j$.
Since $\eta_j = |e_i|$,
$
\Pr\left[ \eta_j \geq |\mu| + ks^2\right] \leq  \Pr\left[\left|e_j - \mu\right| \geq ks^2\right]
$.
By Chebyshev's inequality, 
$
\Pr\left[\left|e_j - \mu\right| \geq ks^2\right] \leq 1/k^2
$.
Combining and 
taking an upper bound, for $k > 1$ we have
$
\Pr\left[\eta_j \geq k|\mu| + ks^2 \right] 
\leq \Pr\left[\eta_j \geq |\mu| + ks^2 \right] 
\leq 1/k^2
$.
\end{proof}

Below, we state and prove the performance bound of the \learnedLL in this model.

 \begin{theorem}
 Consider a LearnedLL data structure using a black box \PMAA with admissible cost $C(\cdot)$.  
 If for each $x_i$, 
 $\tilde{r_i} = r_i + e_i$
 where $e_i$ is drawn independently from a distribution $\calD$ with expectation $\mu$ and variance $s^2$, 
 then  the total element movements 
 over a sequence of at most $n$ insertions 
 is $O\left(n\cdot C((|\mu| + s^2)^2)\right)$
\end{theorem}
\begin{proof}
 In this proof, to simplify the notation we define a variable $\rho = |\mu| + s^2$.  Thus, we want to show a total cost of $O(nC(\rho^2))$.
 
Let $\calP_f$ be the set of actual \PMAs after all inserts have completed.  By Lemma~\ref{lem:finalpmas}, the total number of element movements is 
\[
O\left(\sum_{P\in \calP_f} |P|\cdot C(|P|) \right).
\]

 We split $\calP_f$ into two sets: \PMAs with at most $|P| \leq \rho^2$ assigned ranks (which we call \defn{small \PMAs}), and \PMAs with more than $|P| > \rho^2$ assigned ranks (\defn{large \PMAs}).  That is to say, the number of element movements is:
\[
O\left(
\sum_{\text{small} P\in \calP_f} |P|\cdot C(|P|)
+ 
\sum_{\text{large} P\in \calP_f} |P|\cdot C(|P|)
\right).
\]
For small \PMAs, we can upper bound 

\[
\sum_{\text{small} P\in \calP_f} |P| \cdot C(|P|) = O(n C(\rho^2)).
\]

Now consider large potential \PMAs.
We calculate the probability that any large $P$ has $P\in \calP_F$ by upper bounding the probability that $P$ becomes an actual \PMA; then, we will use this probability to upper bound the expected cost.

If $P\in \calP_f$, then the ranks of the elements stored in $P$ must be a contiguous subset of $\{1, \ldots, n\}$ of size at most $3|P|$.  For $j = \{0, 1, \ldots, \lceil (n -6|P|)/3|P|\rceil\}$, we create a set $R_j = \{3|P|j + 1, \ldots, 3|P|j + 6|P|\}$. 
(Thus, each $R_j$ is a set of size $6|P|$ with its first $3|P|$ elements in common with $R_{j-1}$.)
The ranks of the elements in $P$ must be a subset of $R_{\tau_P}$ for some $\tau_P$.  
There are $|P|$ ranks assigned to $P$; these are also a contiguous subset of $\{1,\ldots, n\}$, so they must be a subset of $R_{\alpha_P}$ for some $\alpha_P$.\footnote{There may be multiple such $\alpha_P$; pick one arbitrarily.}
Let $\delta_P = |\alpha_P - \tau_P|$.

We make two observations; both observations state that for $P$ to exist there must be a high-error element. 
From Lemma~\ref{lem:higherrorelement}, if $P$ is an actual \PMA then there must be an element $x_j$ stored in $P$ with error $\eta_j \geq |P|/2$.  Furthermore, all elements in $P$ have error at least $3|P|(\delta_P - 2)$ (this second observation is important for the case when all items in $P$ have large error).

We split into two cases.  
\paragraph{First case: $\delta_P \geq 3$.} 
In this case,  we use that all items in $P$ must have error at least $3|P|(\delta_P - 2)$.  Let 
$t$ satisfy $|\alpha_P - t| = \delta_P$ (there are two such $t$).
Let $x_j$ be an item with rank in $R_{t}$; by Lemma~\ref{lem:cheby}, the probability that $x_j$ has error at least $3|P|(\delta_P-2)$ is:
 \[
\Pr\left[\eta_j \geq 3|P|(\delta_P - 2)\right] \leq \frac{1}{(3|P|(\delta_P - 2)/\rho)^2}
\]
There are $6|P|$ elements with ranks in $R_{t}$; taking a union bound, the probability that any element has error at least $3|P|(\delta_P-2)$ is at most 
\[
\Pr\left[\exists x_j\in R_{t} \text{ such that } \eta_j \geq 3|P|(\delta_P - 2)\right] \leq
\frac{6|P|\rho^2}{9|P|^2(\delta_P - 2)^2}
= \frac{2\rho^2}{3|P|(\delta_P - 2)^2}.
\]
Taking the union bound over all $\delta_P$, and over both possible values of $t$ for a given $\delta_P$, we obtain the following.  Given that $\delta_P \geq 3$, the probability that $P$ is an actual \PMA is at most (using that $\sum_{x=1}^{\infty} 1/x^2 < 2$):
\[
\sum_{\delta_P = 3}^{n} 
 \frac{4\rho^2}{3|P|(\delta_P - 2)^2}.
< 4\rho^2/|P|.
\]

\paragraph{Second case: $\delta_P \leq 2$.}
By Lemma~\ref{lem:higherrorelement}, if $P\in \calP_f$ there must be some element $x_j$ stored in $P$ with $\eta_j \geq |P|/2$. 
For a given $x_j$, by Lemma~\ref{lem:cheby} we have
\[
\Pr\left[\eta_j \geq |P|/2\right] \leq \frac{4\rho^2}{|P|^2}.
\]

Since $\delta_P \leq 2$, the rank of $x_j$ must be in 
$R_{\tau_P - 2} \cup R_{\tau_P - 1} \cup \ldots \cup R_{\tau_P + 2}$;
there are at most $24|P|$ values in this set.\footnote{We do not attempt to optimize constants in this proof.}
Taking the union bound over all such $x_j$ the probability that $P\in \calP_F$ is at most
 \[
\Pr\left[\exists x_j \text{ such that }\eta_j \geq |P|/2\right] \leq 24|P| 
\cdot\frac{4\rho^2}{|P|^2}.
\leq \frac{96\rho^2}{|P|}.
\]

\paragraph{Putting it all together.} Taking the union bound over the two cases above, the probability that a potential \PMA $P$ 
is an actual \PMA is at most $100\rho^2/|P|$.

Now, we sum the expected cost over all potential \PMAs with $|P| > \rho^2$ assigned ranks.  There are $n/2^{i}$ potential \PMAs with $|P| = 2^i$; if any is an actual \PMA, it has cost $O( 2^{i}C(2^i))$.  From the above, the probability that each is an actual \PMA is $100\rho^2/|P| = 100\rho^2/2^i$.

Therefore, the expected cost of all \PMAs is
\begin{align*}
\sum_{i = 2\lceil \log_2 \rho \rceil}^{\log n} O\left(\frac{n}{2^{i}} \cdot 2^{i}C(2^{i})  \cdot \frac{100\rho^2}{2^i} \right) &=
O\left(n\rho^2 \sum_{i = 2\lceil \log_2 \rho \rceil}^{\log n} C(2^i)/2^i \right) \\
&\leq O\left(n\frac{\rho^2}{2^{2\lceil \log_2 \rho \rceil}} \sum_{j = 0}^{\infty} C(2^{2\lceil \log_2 \rho \rceil + j})/2^j \right) \\
&\leq O\left(nC(\rho^2)\right).
\end{align*}
where the last line is because $C(\cdot)$ is admissible.  Summing the expected cost of $P\in \calP_F$ with height at least $2h_0$ and the cost of $P\in \calP_F$ with height at most $2h_0$ we obtain the lemma.
\end{proof}

\section{Experiments}\label{sec:experiments}

This section presents experimental results on real datasets. The goal is to show the theory is predictive of practice. In particular, the aim is to establish the following. 

\begin{compactitem}
    \item The \LearnedLL improves performance over baseline data structures. Moreover, for the common use case of temporal data, learning from the past leads to future improvements. 
    \item Predictions made on only a small amount of past data  lead to improved performance. 
    \item The \LearnedLL is robust to large errors in the predictions.  
\end{compactitem}
\paragraph{Experimental Setup.} 
We compare the performance of two strong baselines along with our algorithm.  The first baseline is the Packed-Memory Array (PMA)~\cite{BenderDeFa00} and the second is the Adaptive Packed-Memory Array (APMA)~\cite{BenderHu07}. We tested LearnedLLA with both of these \PMAs as the black box \PMA.
For fairness, we seek to ensure all algorithms have the same memory; therefore,  we  implemented the PMA (resp. APMA) as a LearnedLLA using a PMA (resp. APMA) as a black box, such that each element is always placed into the first black box LLA.

We use several datasets from SNAP Large Network Dataset Collection~\cite{snapnets}. All the datasets we use are temporal. The timestamp is used for the arrival order, and an element feature is used for the arriving element's value. 
See Appendix \ref{sec:appendix-experiments} for detailed dataset descriptions.

To generate the predictions for LearnedLLA, we use a contiguous subsequence $L_{\text{train}}$ of the input in temporal order as our training data. Our test data $L_{\text{test}}$ is a contiguous subsequence of the input that comes right after $L_{\text{train}}$, again in temporal order. 

We use two different algorithms for obtaining predictions:

\begin{itemize}
    \item $\text{predictor}_1(L_{\text{train}}, L_{\text{test}})$: For each element $x \in L_{\text{test}}$, this function first finds the rank of $x$ in $L_{\text{train}}$, and then it scales it by $|L_{\text{test}}|/|L_{\text{train}}|$. Finally, it returns these predictions. 
    \item $\text{predictor}_2(L_{\text{train}}, L_{\text{test}})$: Let $a$ be the slope of the best-fit line for points $\{(i,L_{\text{train}}[i])\}_{1 \leq i \leq |L_{\text{train}}|}$. First, this function adds $a \cdot (d + i \cdot (\frac{|L_{\text{test}}|}{|L_{\text{train}}|}-1))$ to the $i$'th element in $L_{\text{train}}$, for each $1\leq i \leq |L_{\text{train}}|$, to obtain $L'_{\text{train}}$, where $d$ is the difference between the starting points of training and test data in our input sequence. Then it returns $\text{predictor}_1(L'_{\text{train}}, L_{\text{test}})$.
\end{itemize}
Let $L^1_{\text{train}}$ and $L^2_{\text{train}}$ be the first and second halves of $L_{\text{train}}$, respectively. To obtain the final predictions for LearnedLLA, we calculate the predictions $P_i:=\text{predictor}_i(L^1_{\text{train}}, L^2_{\text{train}})$ for $i=1,2$, and run LearnedLLA on $L^2_{\text{train}}$ with both of these predictions. If $P_{i^*}$ performs better, we use  $\text{predictor}_{i^*}(L_{\text{train}}, L_{\text{test}})$ as the final predictions for $L_{\text{test}}$.

Our implementation and datasets can be found at \url{https://github.com/AidinNiaparast/LearnedLLA}. See Appendix~\ref{sec:appendix-experiments} for more information about the experimental setup.

\paragraph{Experimental Results.} See Table~\ref{table:exp} for performance on several real datasets.   
We show plots on one dataset (Appendix~\ref{sec:appendix-experiments} contains the plots for the other datasets in Table~\ref{table:exp}). This dataset is Gowalla~\cite{cho2011friendship}, a location-based social networking website where users share their locations by checking in. The latitudes of the locations are used as elements in the input sequence.

Figure~\ref{fig:GowallaLatitude(a)} shows the amortized cost versus the  test data size. For several values for $k$, we use the first and second $n=2^k$ portions of the input as training data and test data, respectively.
Figure~\ref{fig:GowallaLatitude(b)} shows the performance versus the training data size to illustrate how long it takes for the LearnedLLA to learn. The x-axis is the ratio of the training data size to the test data size (in percentage).
Figure~\ref{fig:GowallaLatitude(c)} is a robustness experiment showing performance versus the noise added to predictions, using half of the data as training. In this experiment, we first generate predictions by the algorithm described above, and then we sample $t$ percent of the predictions uniformly at random and make their error as large as possible (the predicted rank is modified to 1 or $n$, whichever is farthest from the current calculated rank). We repeat the experiment five times, each time resampling the dataset, and report the mean and standard deviation of the results of these experiments.

\captionsetup[table]{skip=10pt}

\begin{table}[ht!]
\centering
 \begin{tabular}{c|c c c c c} 
 \textbf{} & \vtop{\hbox{\strut \ \ \ Gowalla}\hbox{\strut \ \ (Latitude)}} & \vtop{\hbox{\strut \ \ \ Gowalla}\hbox{\strut (LocationID)}} 
 & MOOC 
 & AskUbuntu & email-Eu-core
 \\ [0.5ex] 
 \hline
 \rule{0pt}{2.5ex}
 PMA & 6.47 & 6.96 & 19.22 & 24.62 & 21.48\\
 APMA & 6.93 & 5.64 & 16.70 & 11.34 & 21.44\\
LearnedLLA + PMA & {\bf 4.32} & 5.99 & {\bf 11.99} & 14.29 & {\bf 16.48}\\
LearnedLLA + APMA & 4.39 & {\bf 5.01} & 12.13 & {\bf 8.53} & 16.49\\
 \end{tabular}
\caption{Amortized cost of \PMAs on several real datasets. Each column corresponds to a dataset (see Appendix~\ref{sec:appendix-experiments} for dataset descriptions). Rows represent the amortized cost of PMA, APMA, and LearnedLLA using a PMA and APMA as a black box, respectively. 
In all cases, we use the first and second $2^{17}=131072$ portions of the dataset as training and test data, respectively.}
\label{table:exp}
\end{table}

\begin{figure}[ht!]
\newcommand\width{4.4cm}
\newcommand\height{2.5cm}
\newcommand\yaxissub{-10pt}
\centering
\begin{subfigure}{0.33\textwidth}
  \centering
  \hspace{\yaxissub}
  \includegraphics[width=\width, height=\height]{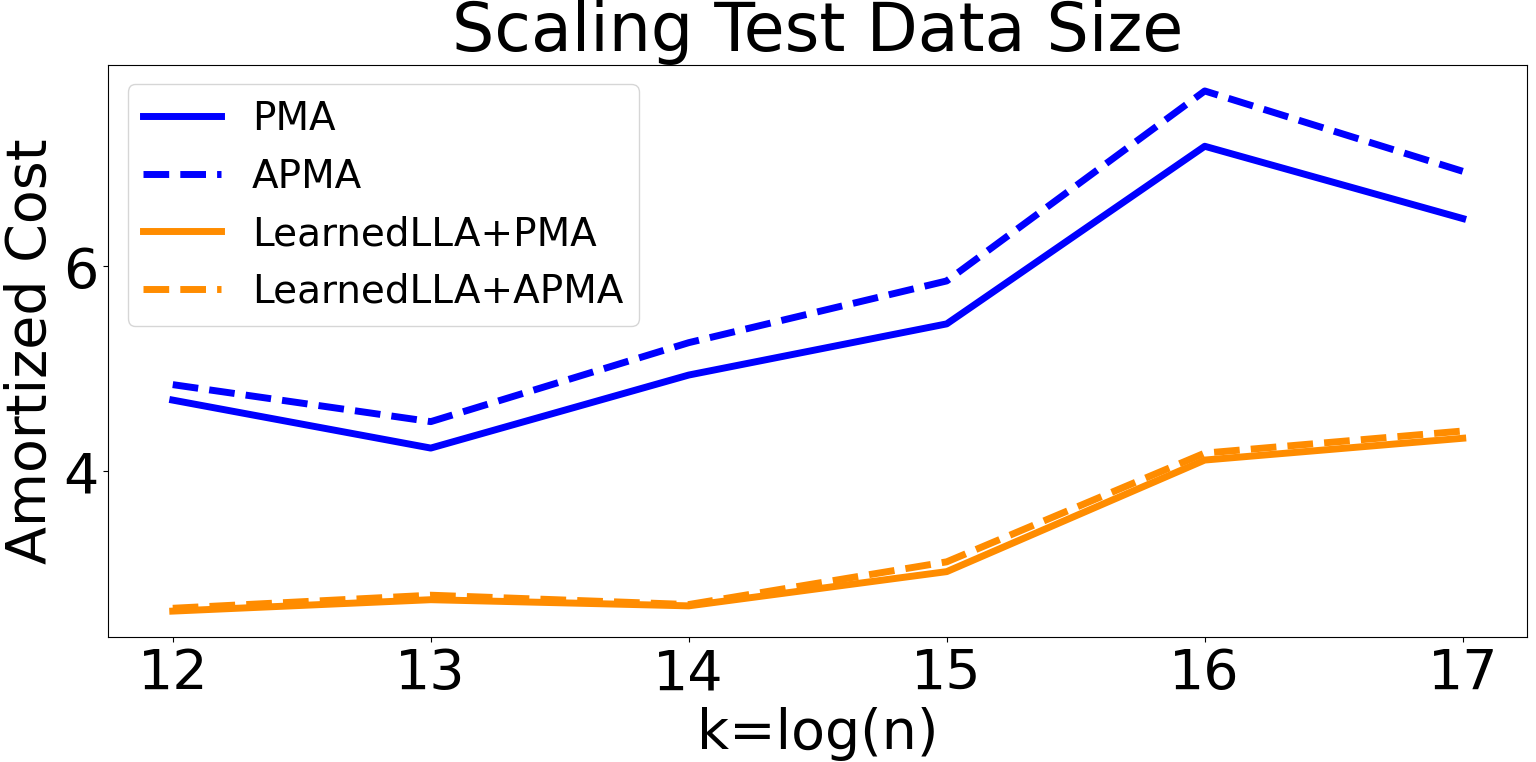}
  \caption{}
  \label{fig:GowallaLatitude(a)}
\end{subfigure}
\begin{subfigure}{0.33\textwidth}
  \centering
  \hspace{\yaxissub}
  \includegraphics[width=\width, height=\height]{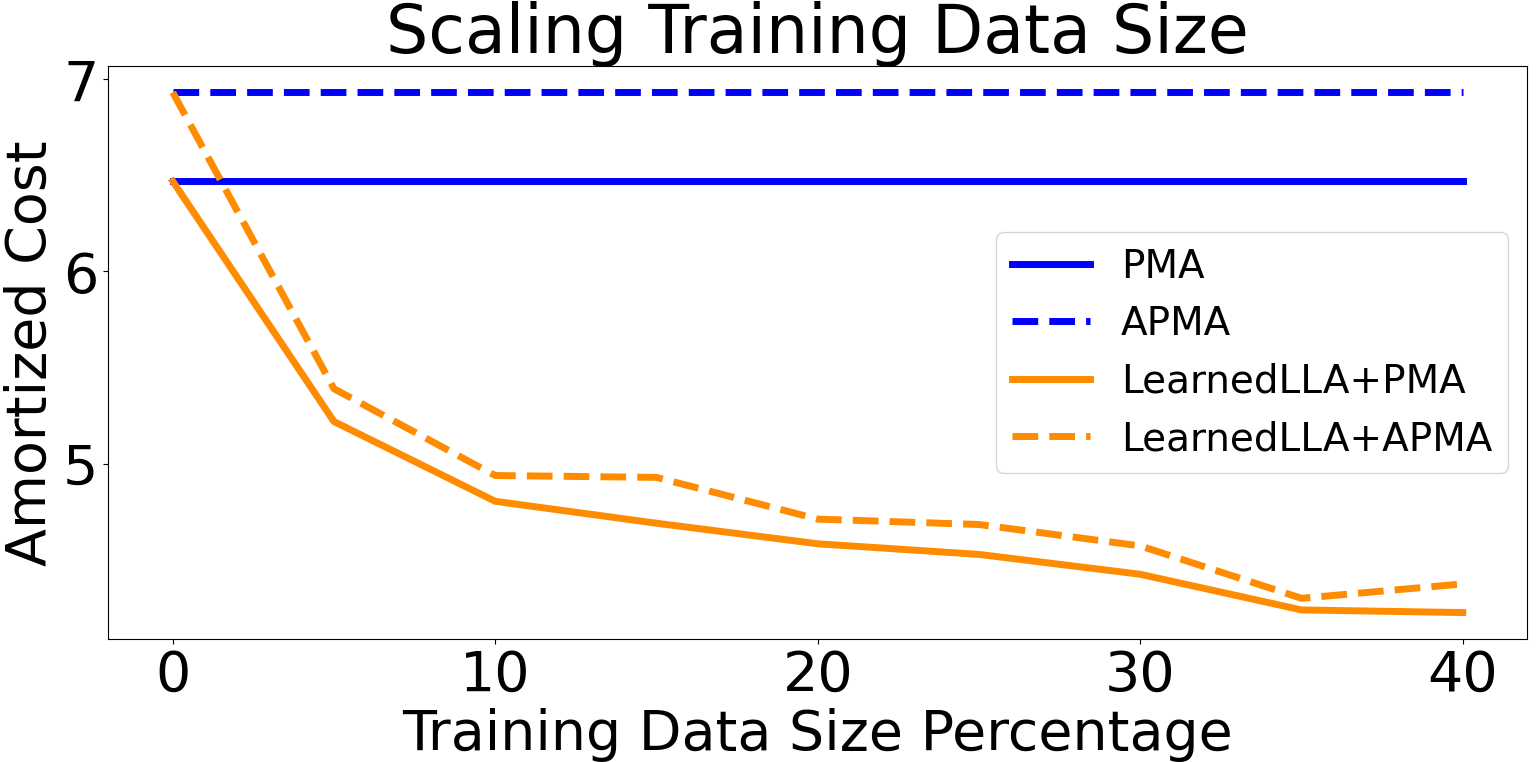}
  \caption{}
  \label{fig:GowallaLatitude(b)}
\end{subfigure}
\begin{subfigure}{0.3\textwidth}
  \centering
  \hspace{\yaxissub}
  \includegraphics[width=\width, height=\height]{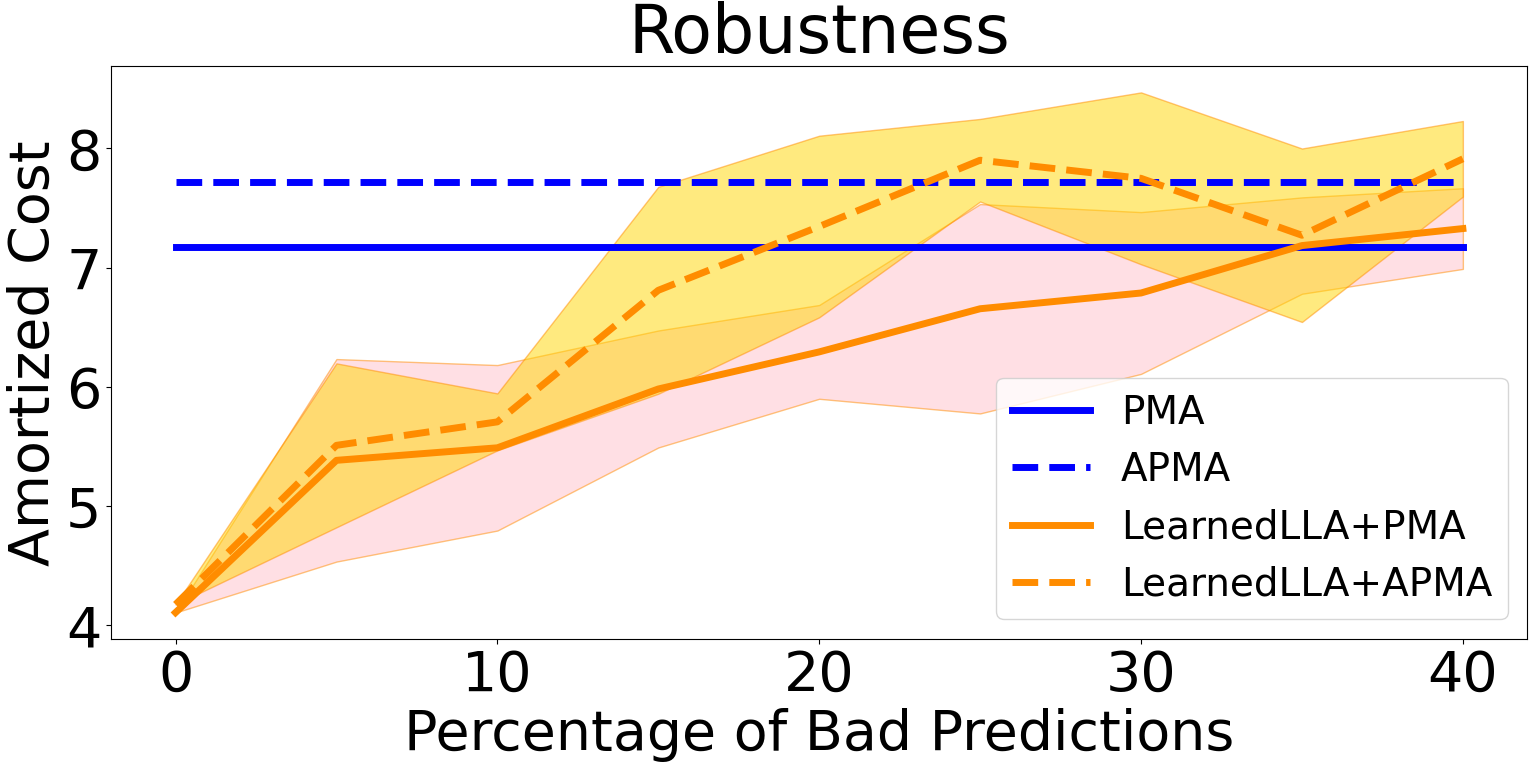}
  \caption{}
  \label{fig:GowallaLatitude(c)}
\end{subfigure}
\caption{Latitudes in the Gowalla dataset. In Fig.~\ref{fig:GowallaLatitude(a)}, we use the first and second $n=2^k$ portions of the input as training data and test data (resp.) for several values of $k$. In Fig.~\ref{fig:GowallaLatitude(b)}, the test data is a fixed portion of the input of size $n=2^{17}=131072$, and we increase the size of the training data, which always comes right before the test data. In Fig.~\ref{fig:GowallaLatitude(c)}, training data and test data are the first and second $2^{16}$ portions of the input (resp.) and we increase the percentage of ``bad'' predictions. For both LearnedLLA algorithms in this plot, we resample the bad predictions 5 times. The lines are the means of these experiments, and the clouds around them show the standard deviation.}
\label{fig:GowallaLatitude}
\end{figure}

\paragraph{Discussion.} Results  in Table~\ref{table:exp}  show that in most cases, using even a simple algorithm to  predict ranks can lead to significant improvements in the performance over the baselines;  in some cases by over 40\%.
In all cases, the \learnedLL improves upon the performance of the black-box LLA used. 
 Furthermore, as illustrated in Figure~\ref{fig:GowallaLatitude(b)},  a small amount of past data---in some cases as small as 5\%---is needed to see a significant separation
between the performance of our method and the baseline LLAs. Finally, Figure~\ref{fig:GowallaLatitude(c)}  suggests that our algorithm is robust to bad predictions. In particular, in this experiment, the maximum error is as large as possible and a significant fraction of the input have errors in their ranks, yet the LearnedLLA is still able to improve over baseline performance.  We remark that when an enormous number of predictions are completely erroneous, the method can have performance worse than baselines.  

We describe more experiments in Appendix~\ref{sec:appendix-experiments}; these experiments further support our conclusions.
\else
 \input{learnedpma}
 \input{experiments}
\fi
\section{Conclusion}\label{sec:conclusion}
In this paper, we show how to use learned predictions to build a learned list labeling array.  We analyze our data structure using the learning-augmented algorithms model. This is the first application of the model to bound the theoretical performance of a data structure. We show that the new data structure  optimally makes use of the predictions. Moreover, our experiments establish that the theory is predictive of practical performance on real data.

An exciting line of work is to determine what other data structures can have improved theoretical performance using predictions.  A feature of the list labeling problem that makes it amenable to the learning-augmented algorithms model is that its cost function and online nature is similar to the competitive analysis model, where predictions have been applied successfully to many problems.  Other data structure problems with similar structure are natural candidates to consider. 

\begin{ack}
Samuel McCauley was supported in part by NSF CCF 2103813.  
Benjamin Moseley was supported in part by a Google Research Award, an Infor Research Award, a Carnegie Bosch Junior Faculty Chair, National Science Foundation grants CCF-2121744 and CCF-1845146 and U. S. Office of Naval Research grant N00014-22-1-2702. Aidin Niaparast was supported in part by U. S. Office of Naval Research under award number N00014-21-1-2243 and the Air Force Office of Scientific Research under award number FA9550-20-1-0080. Shikha Singh was supported in part by NSF CCF 1947789.
\end{ack}


\bibliographystyle{plain}

{\small
\bibliography{pma}
}
\pagebreak
\iffull
\appendix

\section{Appendix}
\label{sec:appendix}

\subsection{Background:  Packed-Memory Arrays}

This section provides background on a common list labeling data structure, \defn{the packed-memory array (PMA)}, introduced by Bender et al.~\cite{BenderDeFa00}. 
The data structure in~\cite{BenderDeFa00} extends the LLAs of~\cite{ItaiKoRo81, Willard82Maintaining, Willard86Good} by adding lower-density thresholds.  Lower-density thresholds provide the added guarantee that any two elements in the array are $\Theta(1)$ slots apart.  
Both versions guarantee an amortized cost of $O(\log^ 2 n)$.  
For simplicity, we describe the PMA using only upper density thresholds.  

\paragraph{Classic PMA.}  Consider an array with $m=cn$ slots, for a constant $c>1$.  Divide the array into ranges of size $\Theta(\log n)$.  These ranges form the \defn{leaves} of an implicit binary tree on top of the $\Theta(n/\log n)$ leaves. Each node in this implicit binary tree corresponds to a subarray containing all leaf ranges within its subtree.

For a node in the implicit binary tree, define the \defn{size($u$)} as number of slots in its subarray and \defn{density($u$)} as the number of elements in its subarray divided by size($u$).  Let the \defn{depth} of the root node be $0$, nodes and leaf nodes be at depth $d = \Theta(\log n/\log \log n)$.  
  
For each node $u$ at depth $k$, let $\tau_k$ be the {\bf density threshold}.  Let $\tau_0$ and $\tau_d$ be constants such that $0 < \tau_0 < \tau_d < 1$. For a node at depth $k$, let $\tau_k = \tau_0 + (\tau_d - \tau_0) \cdot \frac kd$.

A node $u$ is \defn{within threshold} if $\text{density}(u) \leq \tau_k$.  If a node is within threshold, then all its descendants are also within threshold.  

To insert an element $x$, determine the leaf range $r$ it belongs to.  If the leaf range $r$ is within threshold, there is always a slot for $x$.  Insert $x$ in its slot and rebalance the $r$ by evenly distributing all elements.  If the leaf range is out of threshold, find the first ancestor $r_a$ of the $r$ that is within threshold and rebalance all elements evenly in that range.  Now $r$ is within threshold and there is room to insert $x$. 

To see why the amortized cost of insertion is $O(\log^2 n)$, consider the insertion of element $x$ that causes a node $u$ at depth $k$ to be rebalanced.  Then, a child $v$ of $u$ must be out of threshold, that is, $\text{density}(v) > \tau_{k+1}$.  After we rebalance $u$, both $v$ (and its sibling) have density at most $\tau_k$ (density threshold of its parent $u$).  The node $u$ would need to rebalanced again when either $v$ (or its sibling) go out of treshold again.  This requires at least $(\tau_{k+1} - \tau_k) \cdot \text{size}(v)$ additional insertions.  A rebalance of node $u$ costs $\text{size}(u)$.
Thus, the amortized cost of rebalancing $u$ is:
\begin{align*}
\frac{\text{size}(u)}{(\tau_{k+1} - \tau_k) \text{size}(v)} = \frac{2}{\tau_{k+1} - \tau_k} = \frac{2d}{ \tau_d - \tau_0} =O(\log n)
\end{align*}

Since an insertion can contribute to $O(\log n)$ ancestors being out of balance, the overall amortized cost of an insertion is $O(\log^2 n)$.

\subsection{Additional Experiments}\label{sec:appendix-experiments}

In this section, we further describe the experimental setup and the datasets we use. We also present more experimental results.

\paragraph{Experimental setup.} We use a machine with 11th Gen Intel Core i7 CPU 2.80GHz, 32GB of RAM, 128GB NVMe KIOXIA disk drive, and running 64-bit Windows 10 Enterprise to run our experiments. We remark that amortized cost,  the average number of element movements,  is hardware-independent. 
We use the following density thresholds for the PMA and APMA: root's lower threshold: 0.2,
leaves' lower threshold: 0.1,
root's upper threshold: 0.5,
leaves' upper threshold: 0.9.
We add a $-\infty$ element at the beginning of each experiment. This is to make sure the internal predictor data structure in APMA operates as expected from the beginning of the experiment (see~\cite{BenderHu07}). 
The datasets we use might include duplicate elements, and we use the same algorithm to insert these elements, even though in Section~\ref{sec:prelim}, we assume the elements in the input sequence form a set. This does not affect any of the algorithms and they are still well-defined. The relative order between  duplicates can be arbitrary. 
In LearnedLLA, when we insert an element $x$, to find the black box \PMAs containing the predecessor and successor of $x$, we use Python Sorted Containers library\footnote{https://grantjenks.com/docs/sortedcontainers/} (note that since we only measure the amortized cost, our results are independent of the function used to find these \PMAs).
For measuring the amortized cost in the experiments, we do not count the first assignment of a label to an element as a relabel (note that this is in contrast to the theory section of the paper).

\captionsetup[sub]{font=small}
\begin{figure}[ht]
\ifabstract
\vspace{-.05in}
\fi
\newcommand\width{4.4cm}
\newcommand\height{2.5cm}
\newcommand\yaxissub{-10pt}
\centering
\begin{subfigure}{0.33\textwidth}
  \centering
  \hspace{\yaxissub}
  \includegraphics[width=\width, height=\height]{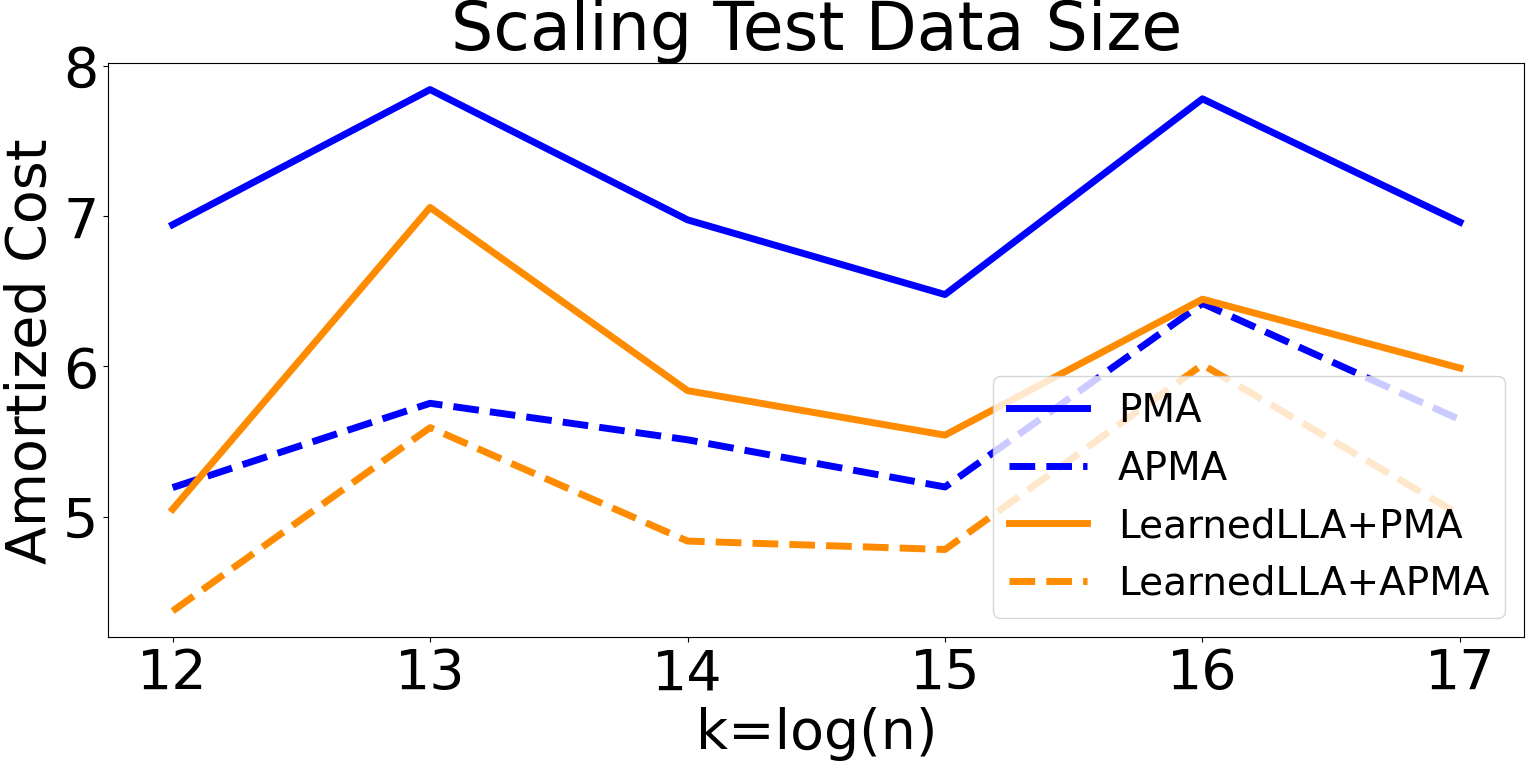}
  \caption{}
  \label{fig:GowallaLocationID(a)}
\end{subfigure}
\begin{subfigure}{0.33\textwidth}
  \centering
  \hspace{\yaxissub}
  \includegraphics[width=\width, height=\height]{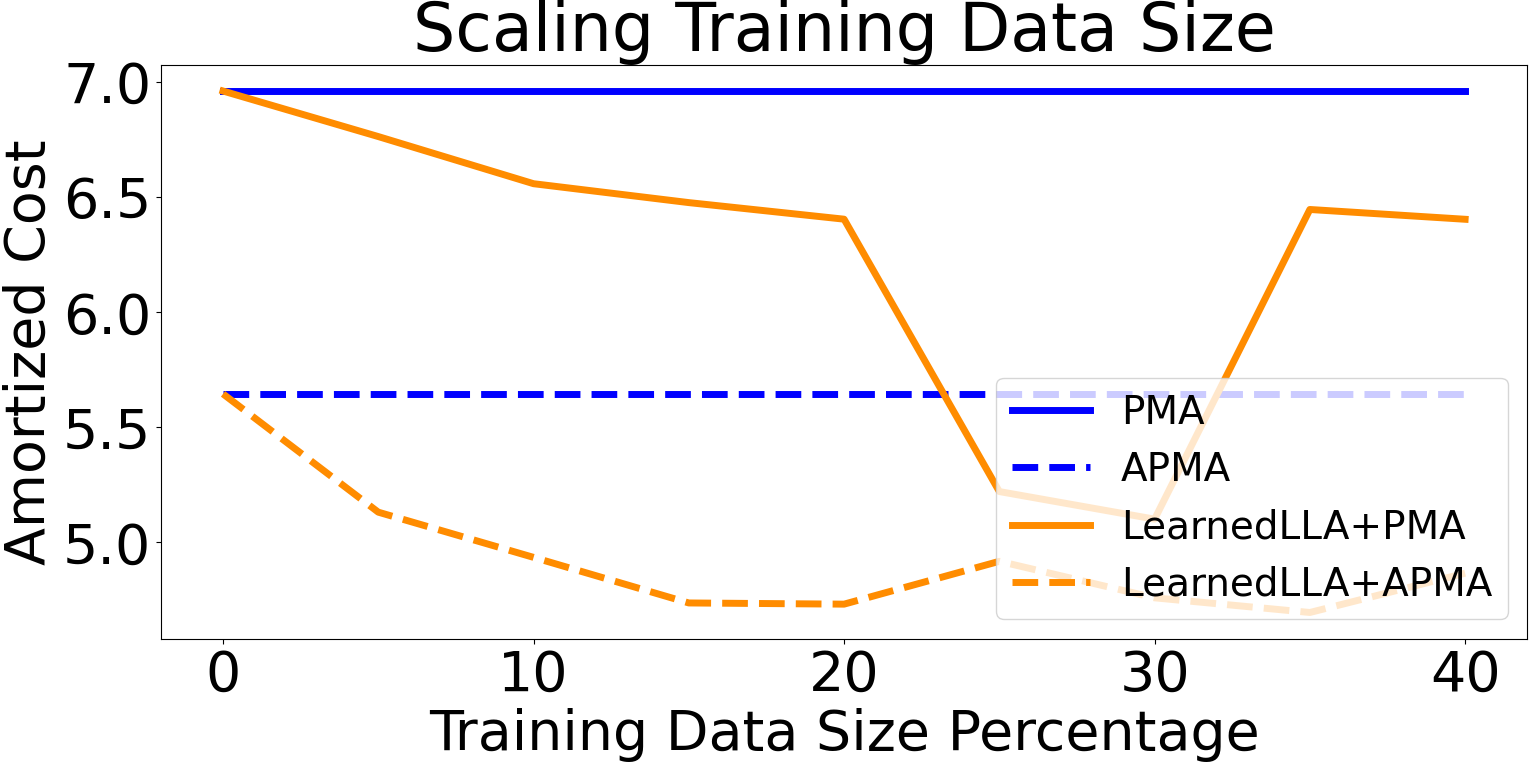}
  \caption{}
  \label{fig:GowallaLocationID(b)}
\end{subfigure}
\begin{subfigure}{0.3\textwidth}
  \centering
  \hspace{\yaxissub}
  \includegraphics[width=\width, height=\height]{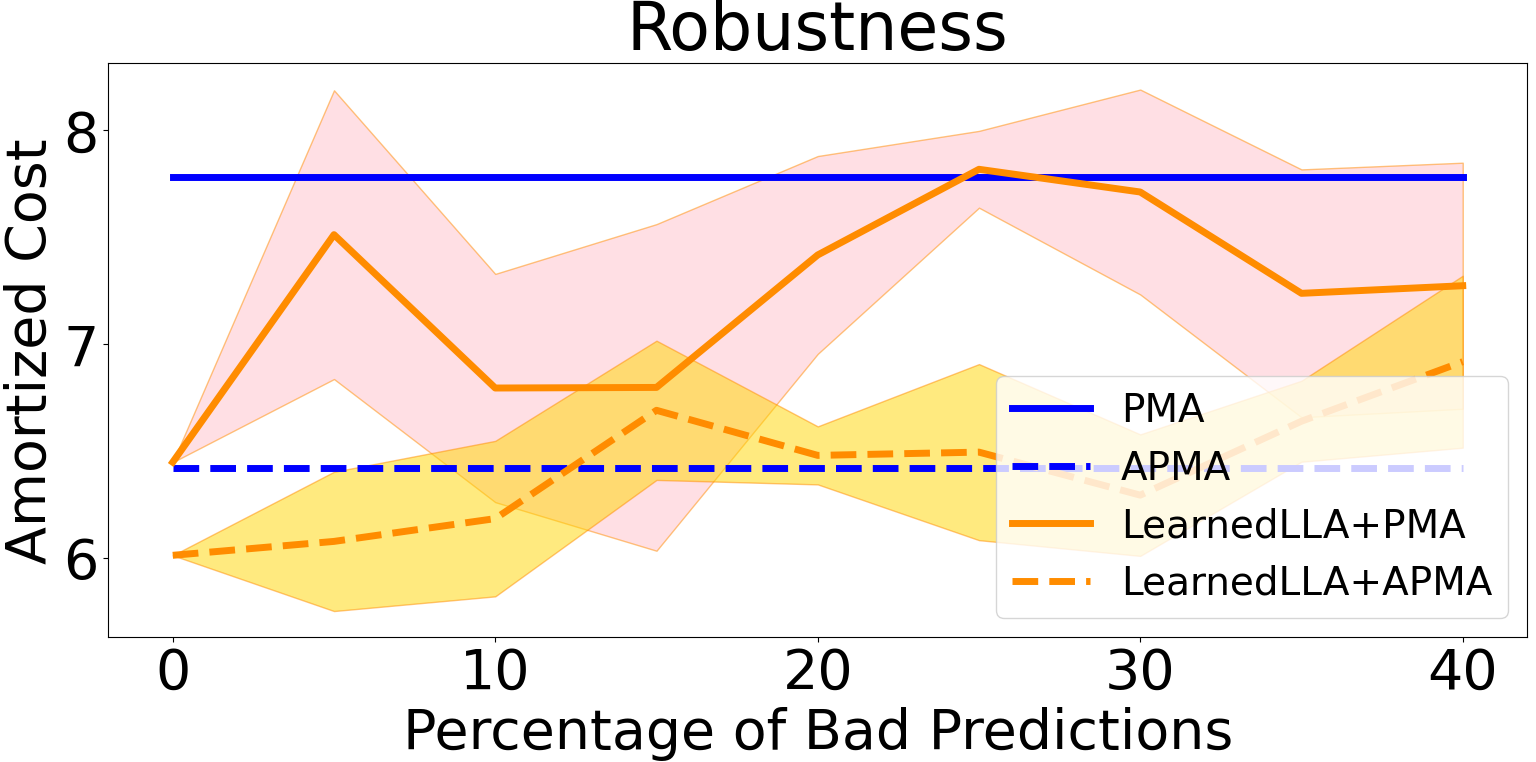}
  \caption{}
  \label{fig:GowallaLocationID(c)}
\end{subfigure}
\caption{Gowalla (LocationID)}
\label{fig:GowallaLocationID}
\end{figure}

\begin{figure}[ht]
\ifabstract
\vspace{-.05in}
\fi
\newcommand\width{4.4cm}
\newcommand\height{2.5cm}
\newcommand\yaxissub{-10pt}
\centering
\begin{subfigure}{0.33\textwidth}
  \centering
  \hspace{\yaxissub}
  \includegraphics[width=\width, height=\height]{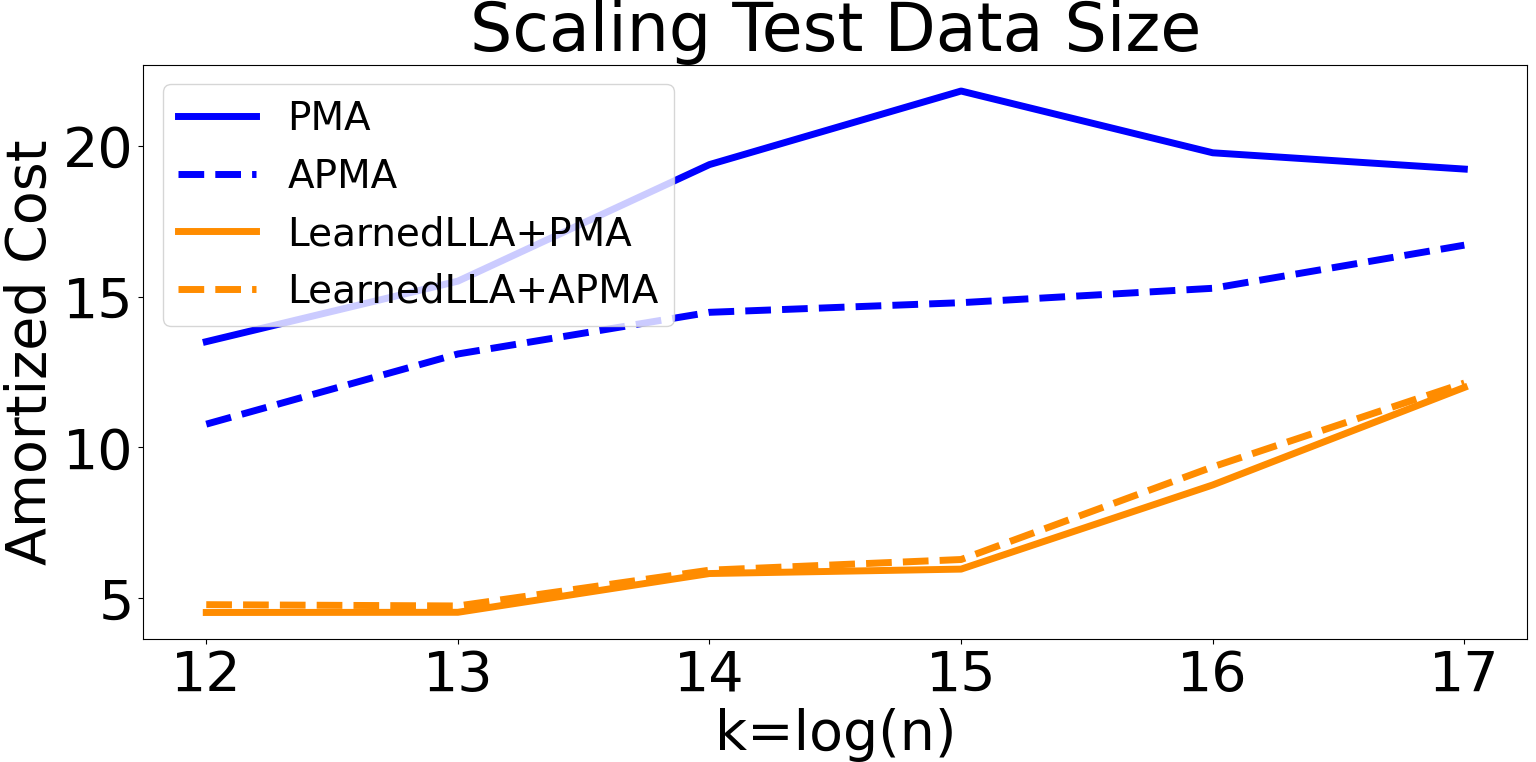}
  \caption{}
  \label{fig:MOOC(a)}
\end{subfigure}
\begin{subfigure}{0.33\textwidth}
  \centering
  \hspace{\yaxissub}
  \includegraphics[width=\width, height=\height]{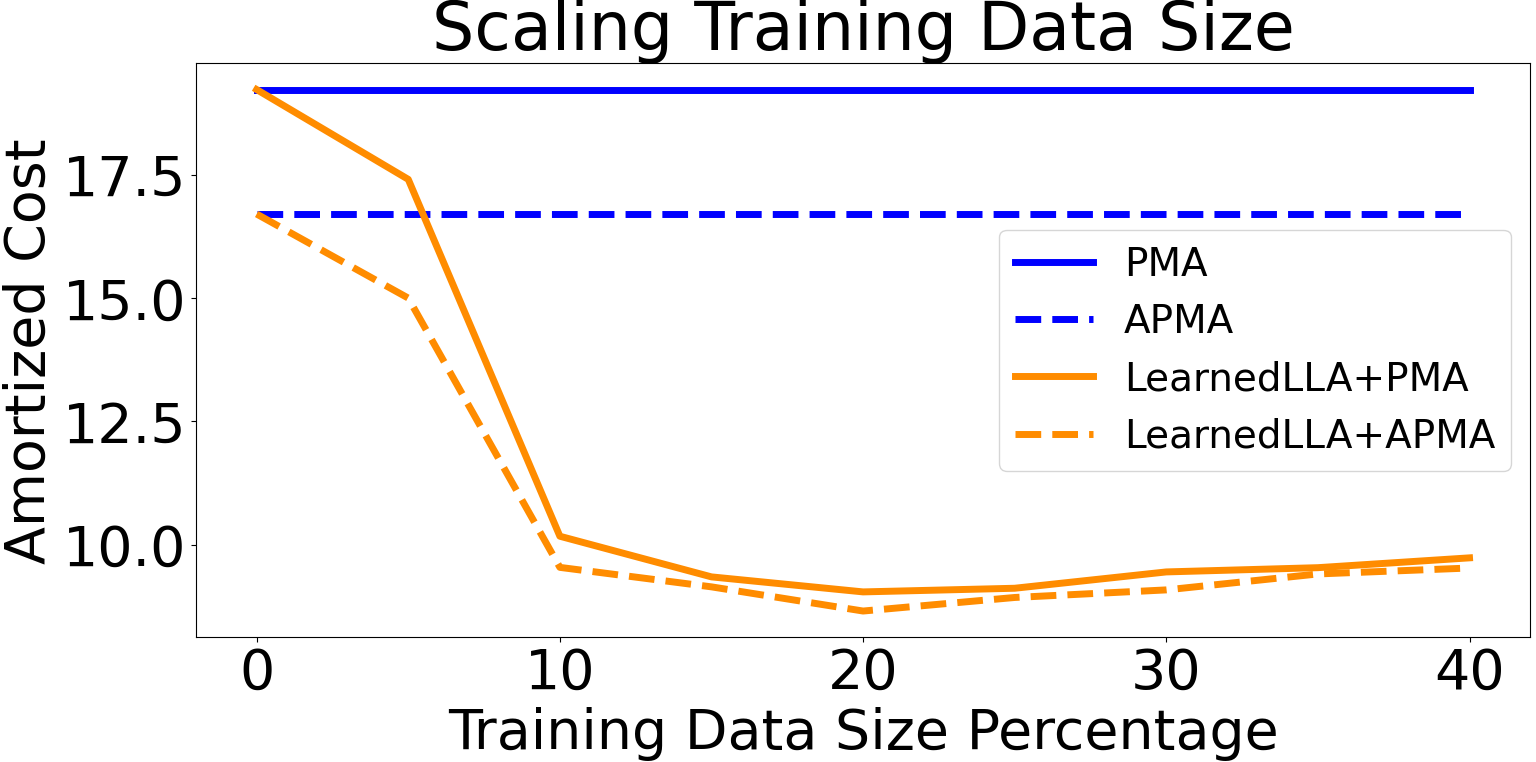}
  \caption{}
  \label{fig:MOOC(b)}
\end{subfigure}
\begin{subfigure}{0.3\textwidth}
  \centering
  \hspace{\yaxissub}
  \includegraphics[width=\width, height=\height]{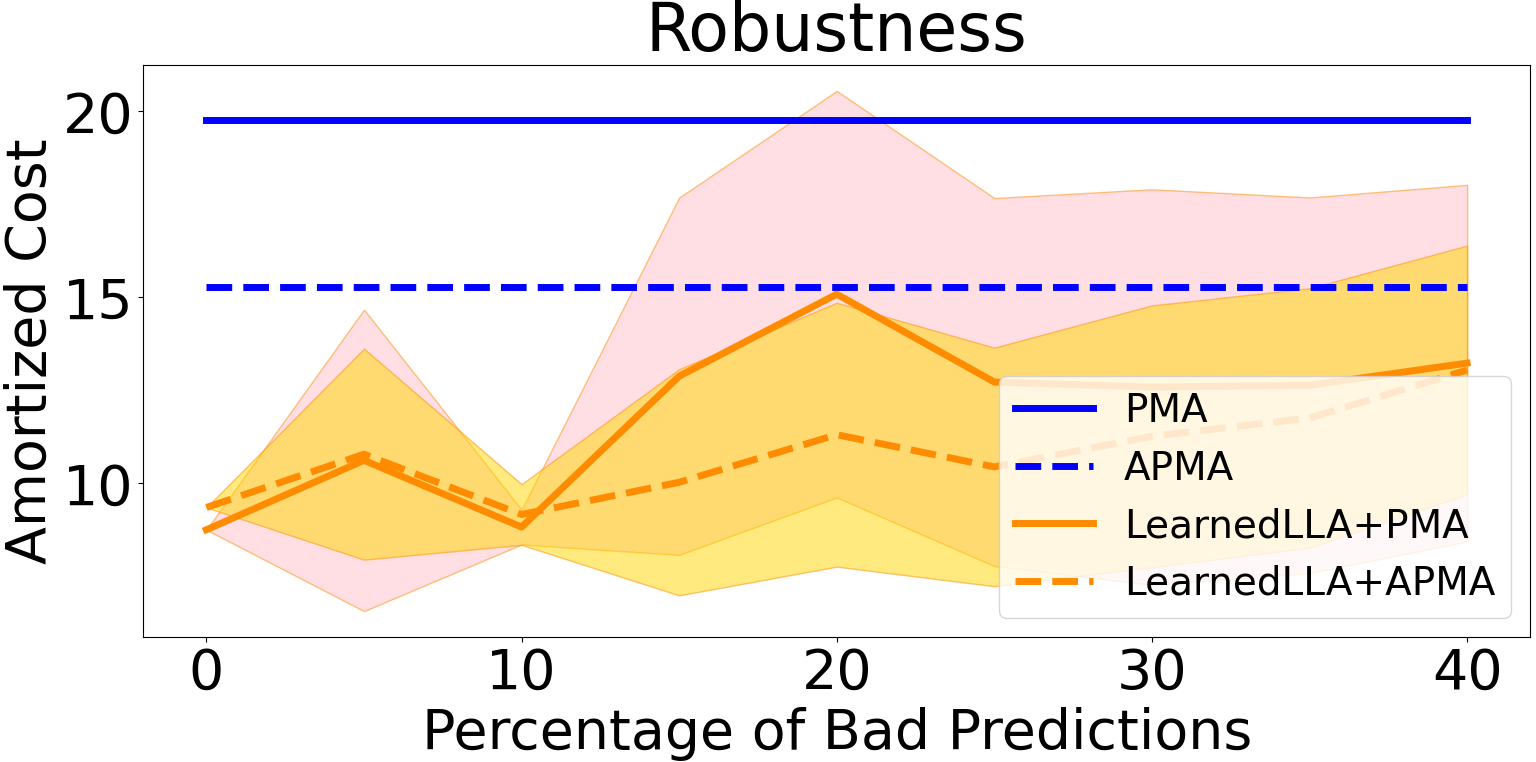}
  \caption{}
  \label{fig:MOOC(c)}
\end{subfigure}
\caption{MOOC}
\label{fig:MOOC}
\end{figure}

\begin{figure}[ht]
\ifabstract
\vspace{-.05in}
\fi
\newcommand\width{4.4cm}
\newcommand\height{2.5cm}
\newcommand\yaxissub{-10pt}
\centering
\begin{subfigure}{0.33\textwidth}
  \centering
  \hspace{\yaxissub}
  \includegraphics[width=\width, height=\height]{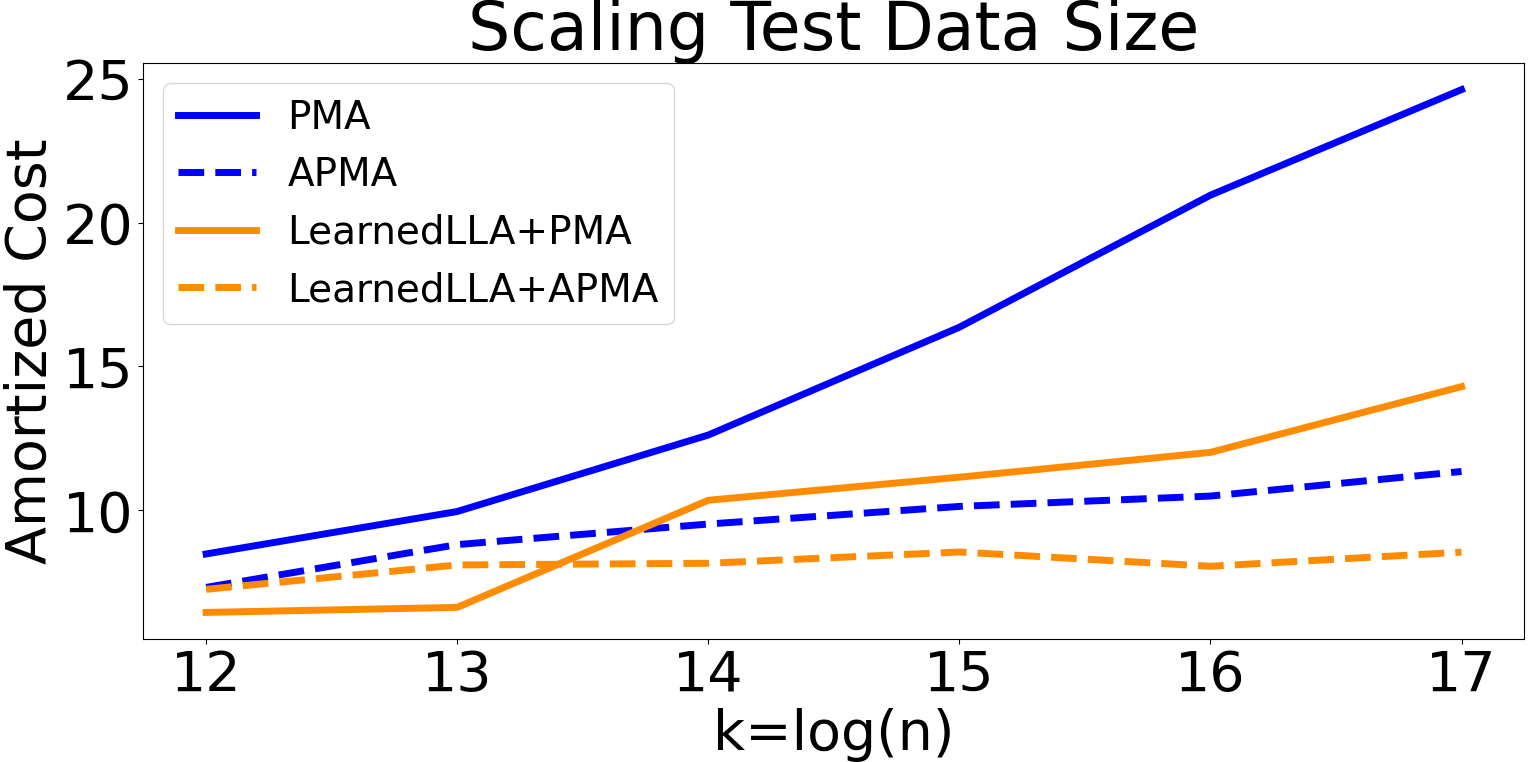}
  \caption{}
  \label{fig:AskUbuntu(a)}
\end{subfigure}
\begin{subfigure}{0.33\textwidth}
  \centering
  \hspace{\yaxissub}
  \includegraphics[width=\width, height=\height]{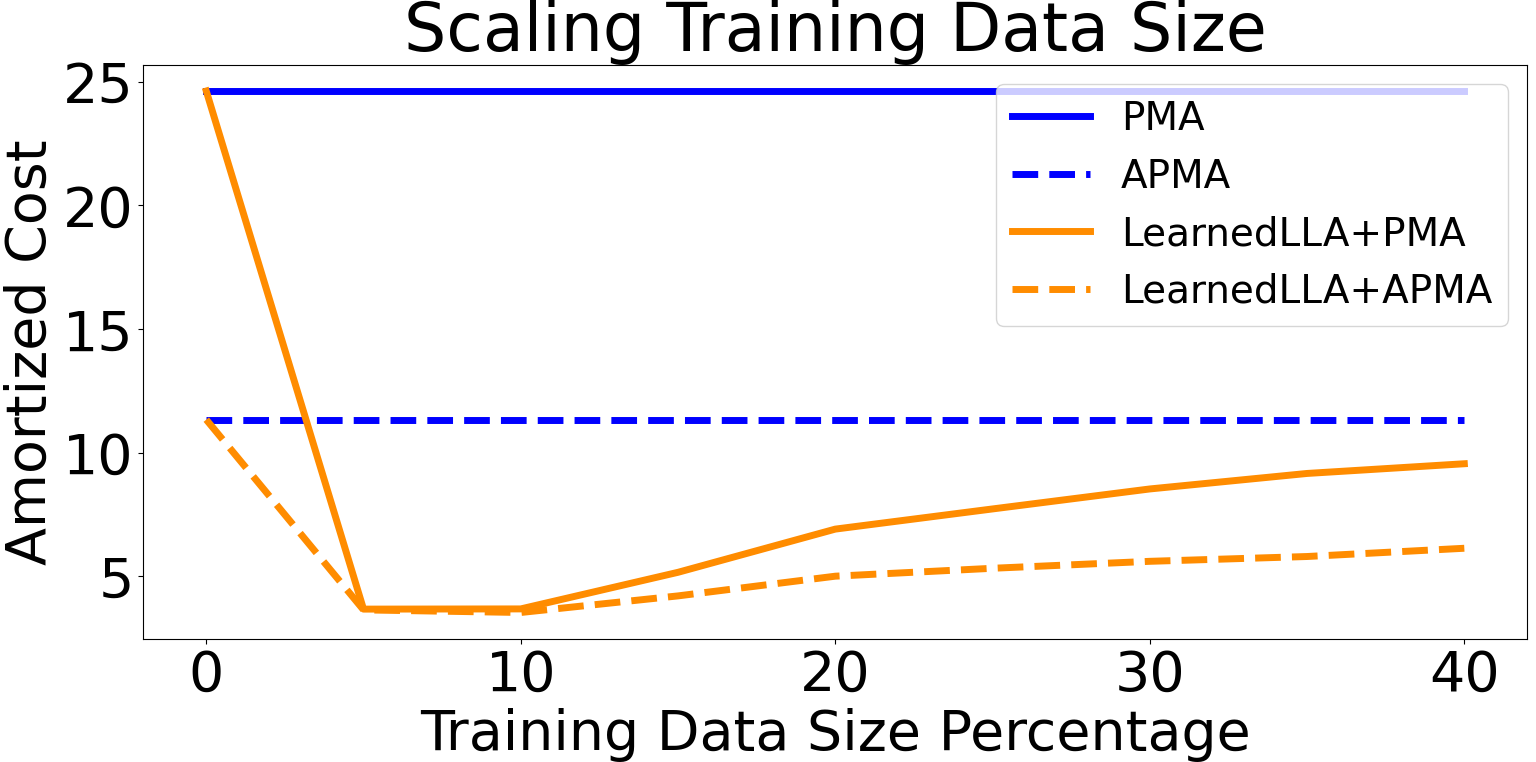}
  \caption{}
  \label{fig:AskUbuntu(b)}
\end{subfigure}
\begin{subfigure}{0.3\textwidth}
  \centering
  \hspace{\yaxissub}
  \includegraphics[width=\width, height=\height]{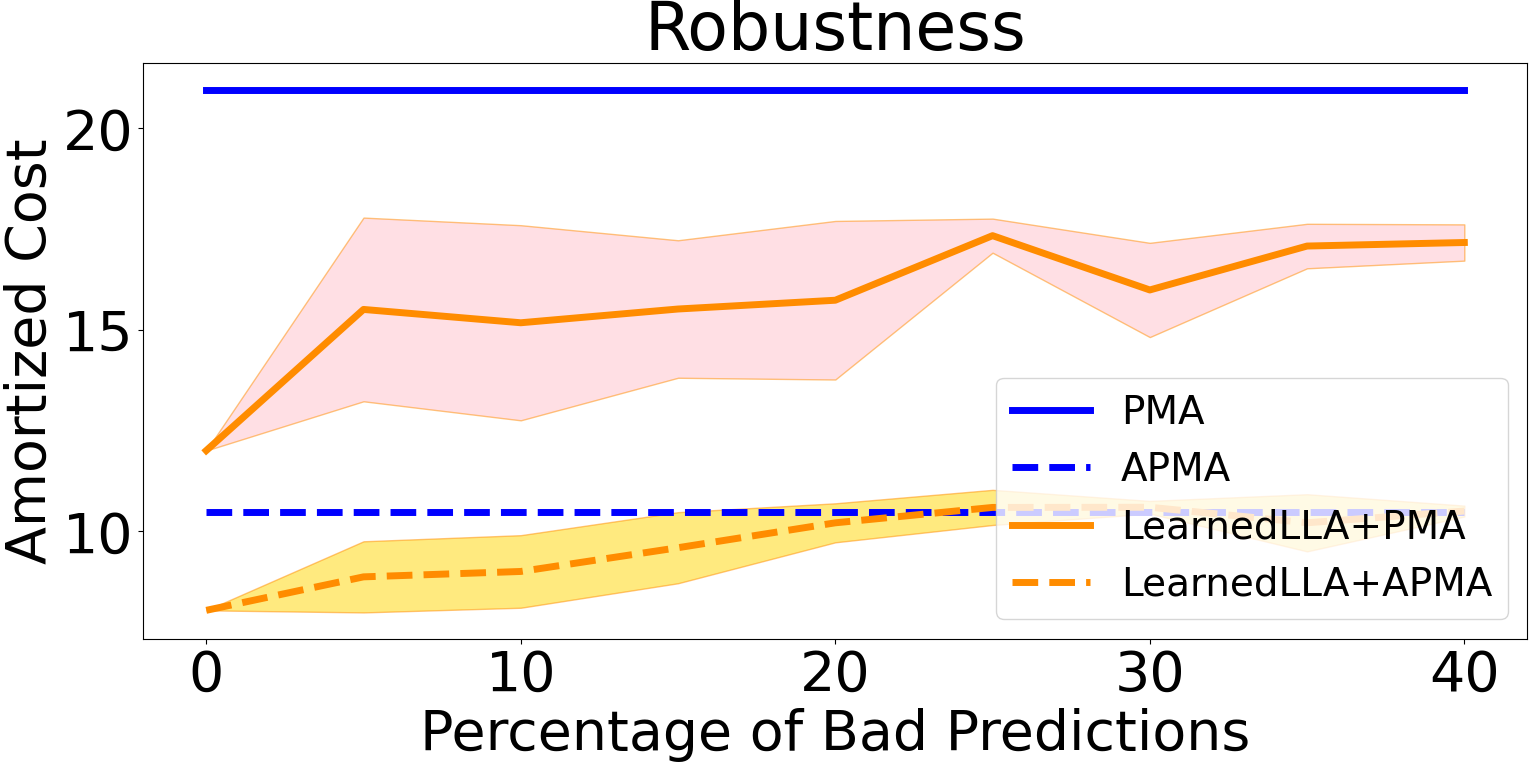}
  \caption{}
  \label{fig:AskUbuntu(c)}
\end{subfigure}
\caption{AskUbuntu}
\label{fig:AskUbuntu}
\end{figure}

\begin{figure}[ht]
\ifabstract
\vspace{-.05in}
\fi
\newcommand\width{4.4cm}
\newcommand\height{2.5cm}
\newcommand\yaxissub{-10pt}
\centering
\begin{subfigure}{0.33\textwidth}
  \centering
  \hspace{\yaxissub}
  \includegraphics[width=\width, height=\height]{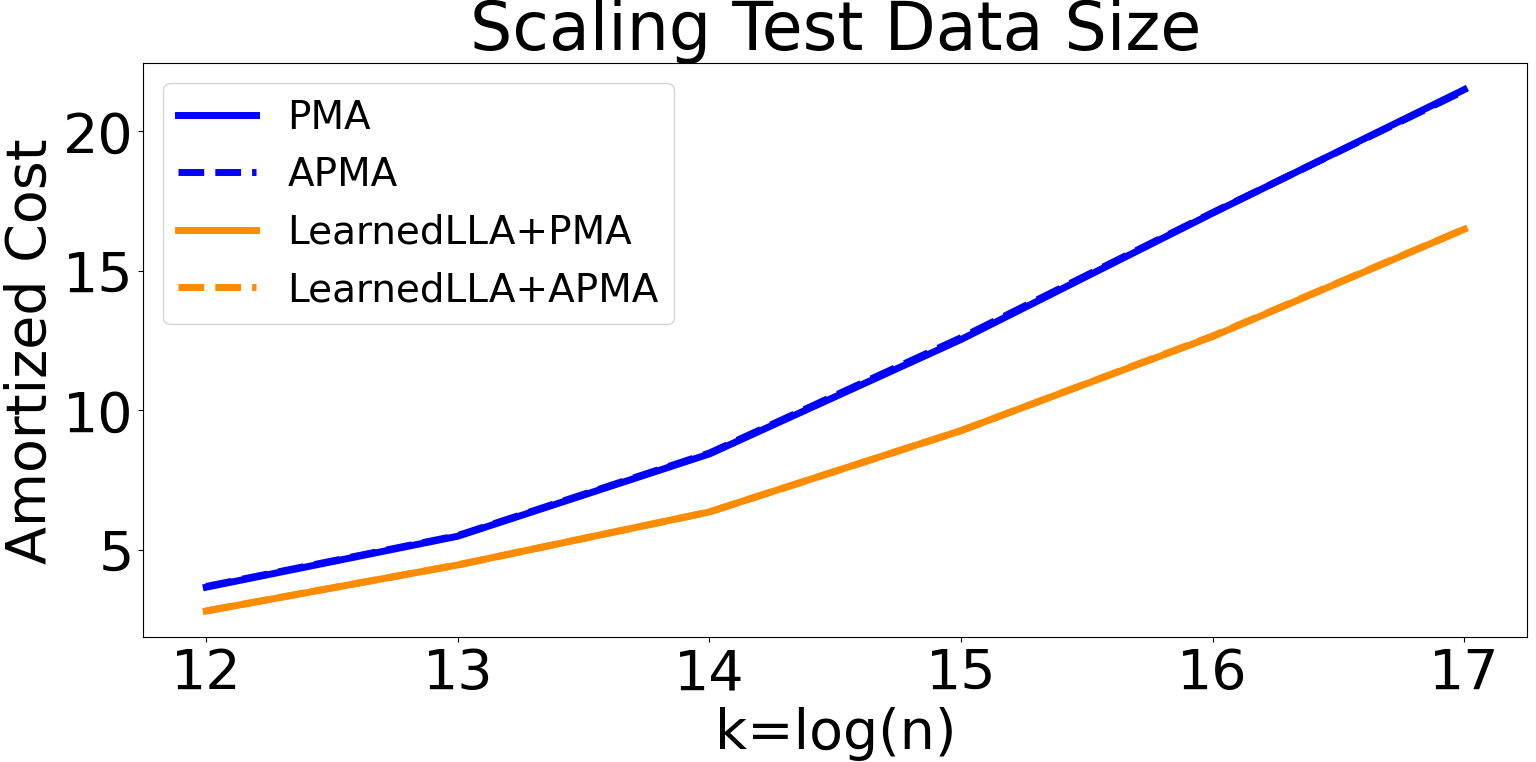}
  \caption{}
  \label{fig:email-Eu-core(a)}
\end{subfigure}
\begin{subfigure}{0.33\textwidth}
  \centering
  \hspace{\yaxissub}
  \includegraphics[width=\width, height=\height]{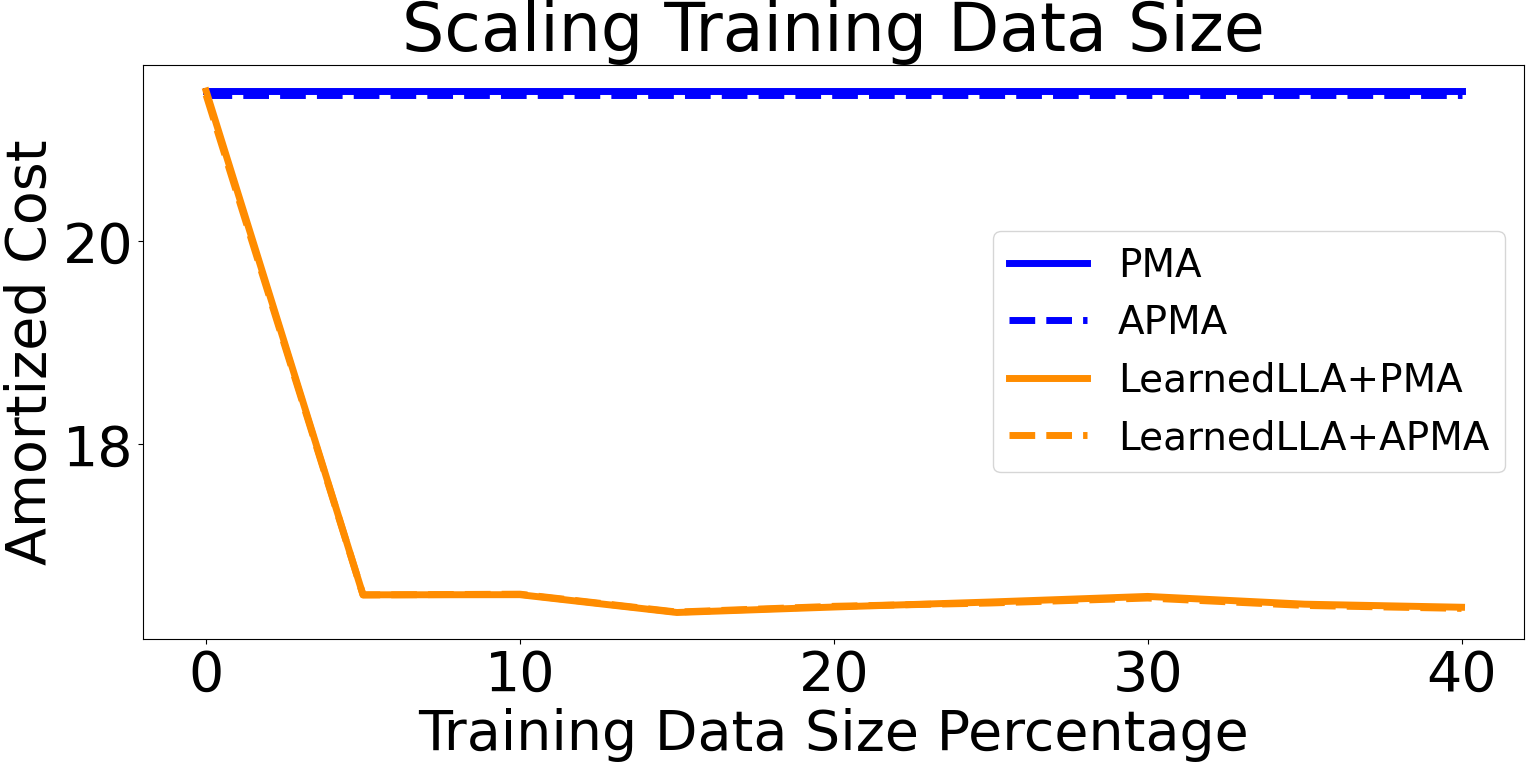}
  \caption{}
  \label{fig:email-Eu-core(b)}
\end{subfigure}
\begin{subfigure}{0.3\textwidth}
  \centering
  \hspace{\yaxissub}
  \includegraphics[width=\width, height=\height]{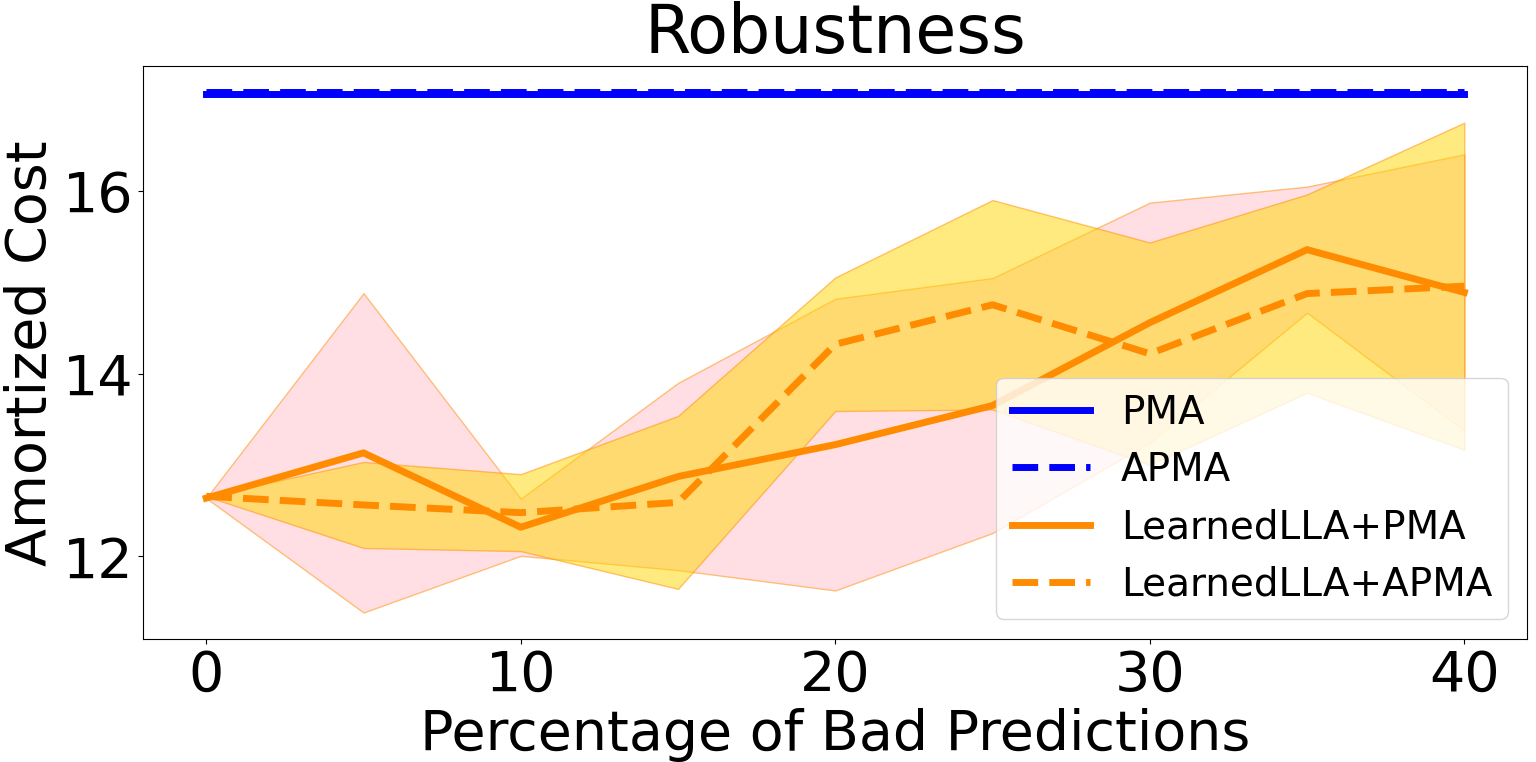}
  \caption{}
  \label{fig:email-Eu-core(c)}
\end{subfigure}
\caption{email-Eu-core}
\label{fig:email-Eu-core}
\end{figure}

\paragraph{Dataset description.} 
Here we describe the real temporal datasets we use in our experiments. In all cases, we use a prefix of the dataset in temporal order as the input sequence.
\begin{itemize}
    \item Gowalla\footnote{https://snap.stanford.edu/data/loc-Gowalla.html}~\cite{cho2011friendship}: Gowalla is a location-based social networking website where users share their locations by checking in. We use the location ID and latitude of the users that check in.
    \item MOOC\footnote{https://snap.stanford.edu/data/act-mooc.html}~\cite{kumar2019predicting}: The MOOC user action dataset represents the actions taken by users on a popular MOOC platform. The actions are represented as a directed, temporal network. The nodes represent users and course activities (targets), and edges represent the actions by users on the targets. We use the user IDs as our input sequence.
    \item AskUbuntu\footnote{https://snap.stanford.edu/data/sx-askubuntu.html}: This is a temporal network of interactions on the stack exchange web site Ask Ubuntu. There are three different types of interactions represented by a directed edge $(u, v, t)$: i. user $u$ answered user $v$'s question at time $t$,
ii. user $u$ commented on user $v$'s question at time $t$, and 
 iii. user $u$ commented on user $v$'s answer at time $t$. We use the IDs of target users in the answers-to-questions network as the input sequence.
    \item email-Eu-core\footnote{https://snap.stanford.edu/data/email-Eu-core-temporal.html}~\cite{paranjape2017motifs} The network was generated using email data from a large European research institution. The e-mails only represent communication between institution members (the core), and the dataset does not contain incoming messages from or outgoing messages to the rest of the world. A directed edge $(u, v, t)$ means that person $u$ sent an e-mail to person $v$ at time $t$. A separate edge is created for each recipient of the e-mail. We use the IDs of target users as our input sequence. 
\end{itemize}

\paragraph{Results.} In~\Cref{fig:GowallaLocationID,fig:MOOC,fig:AskUbuntu,fig:email-Eu-core}, we show plots for other datasets in Table~\ref{table:exp}. The setup is exactly similar to Figure~\ref{fig:GowallaLatitude}.

\paragraph{Discussion.} These results further support our conclusions. Note that in some cases, increasing the size of the training set results in slightly worse performance for LearnedLLA. We believe this is because as we increase the size of the training data, we use older data as training.

\fi

\end{document}